\documentclass[11pt,a4paper]{article}

\usepackage{fancyhdr}
\usepackage{amsmath,mathtools}
\usepackage{bm}
\usepackage{esint}
\usepackage{amsfonts}
\usepackage{amsthm}
\usepackage{amssymb}
\usepackage{amsbsy}
\usepackage{verbatim}
\usepackage{graphicx}
\usepackage{url}
\usepackage{mathrsfs}
\usepackage[hmargin = 1in,vmargin=.75in]{geometry}
\usepackage{color}
\usepackage{amstext}
\usepackage{stmaryrd}
\usepackage{enumerate}
\usepackage{algpseudocode}
\usepackage{algorithm}
\usepackage{pifont}
\usepackage{prettyref}
\usepackage{subcaption}
\font\msbm=msbm10

\numberwithin{equation}{section}

\theoremstyle{plain}
\newtheorem{theorem}{Theorem}[section]
\newtheorem{lemma}[theorem]{Lemma}
\newtheorem{corollary}[theorem]{Corollary}

\newtheorem{proposition}[theorem]{Proposition}

\newtheorem{remark}[theorem]{Remark}
\def\mathbb#1{\hbox{\msbm{#1}}}

\newcommand{\bu}{\boldsymbol{u}}

\newcommand{\bx}{\boldsymbol{x}}

\newcommand{\BA}{\boldsymbol{A}}

\newcommand{\BC}{\boldsymbol{C}}
\newcommand{\BD}{\boldsymbol{D}}
\newcommand{\BE}{\boldsymbol{E}}

\newcommand{\BG}{\boldsymbol{G}}

\newcommand{\BO}{\boldsymbol{O}}

\newcommand{\BQ}{\boldsymbol{Q}}
\newcommand{\BR}{\boldsymbol{R}}
\newcommand{\BS}{\boldsymbol{S}}

\newcommand{\BU}{\boldsymbol{U}}
\newcommand{\BV}{\boldsymbol{V}}
\newcommand{\BW}{\boldsymbol{W}}
\newcommand{\BX}{\boldsymbol{X}}
\newcommand{\BY}{\boldsymbol{Y}}
\newcommand{\BZ}{\boldsymbol{Z}}

\newcommand{\BDelta}{\boldsymbol{\Delta}}

\newcommand{\BLambda}{\boldsymbol{\Lambda}}
\newcommand{\BSigma}{\boldsymbol{\Sigma}}

\newcommand{\PP}{\mathcal{P}}

\newcommand{\I}{\boldsymbol{I}}
\newcommand{\RR}{\mathbb{R}}
\newcommand{\lag}{\langle}
\newcommand{\rag}{\rangle}

\newcommand{\eps}{\epsilon}

\DeclareMathOperator{\Tr}{Tr}

\DeclareMathOperator{\E}{\mathbb{E}}

\DeclareMathOperator{\Od}{O}

\DeclareMathOperator{\rank}{\text{rank}}

\DeclareMathOperator{\argmin}{argmin}
\DeclareMathOperator{\argmax}{argmax}
\renewcommand{\Pr}{\mathbb{P}}
\begin{document}
\title{\bf Near-Optimal Bounds for Generalized Orthogonal Procrustes Problem via Generalized Power Method}
\author{Shuyang Ling\thanks{New York University Shanghai (Email: sl3635@nyu.edu). This work is (partially) financially supported by the National Key R\&D Program of China, Project Number 2021YFA1002800,  National Natural Science Foundation of China (NSFC) No.12001372, Shanghai Municipal Education Commission (SMEC) via Grant 0920000112, and NYU Shanghai Boost Fund.}}

\maketitle

\begin{abstract}

Given multiple point clouds, how to find the rigid transform (rotation, reflection, and shifting) such that these point clouds are well aligned? This problem, known as the generalized orthogonal Procrustes problem (GOPP), has found numerous applications in statistics, computer vision, and imaging science. While one commonly-used method is finding the least squares estimator, it is generally an NP-hard problem to obtain the least squares estimator exactly due to the notorious nonconvexity. In this work, we apply the semidefinite programming (SDP) relaxation and the generalized power method to solve this generalized orthogonal Procrustes problem. In particular, we assume the data are generated from a signal-plus-noise model: each observed point cloud is a noisy copy of the same unknown point cloud transformed by an unknown orthogonal matrix and also corrupted by additive Gaussian noise. We show that the generalized power method (equivalently alternating minimization algorithm) with spectral initialization converges to the unique global optimum to the SDP relaxation, provided that the signal-to-noise ratio is high. Moreover, this limiting point is exactly the least squares estimator and also the maximum likelihood estimator. In addition, we derive a block-wise estimation error for each orthogonal matrix and the underlying point cloud. Our theoretical bound is near-optimal in terms of the information-theoretic limit (only loose by a factor of the dimension and a log factor). Our results significantly improve the state-of-the-art results on the tightness of the SDP relaxation for the generalized orthogonal Procrustes problem, an open problem posed by Bandeira, Khoo, and Singer in 2014.

\end{abstract}

\section{Introduction}

Suppose we observe $n$ noisy copies $\{\BA_i\}_i^n$ of a common unknown point cloud $\BA\in\RR^{d\times m}$ transformed by unknown orthogonal matrices $\{\BO_i\}_{i=1}^n$,
\[
\BA_i = \BO_i\BA + \sigma\BW_i, \quad 1\leq i\leq n,
\]
where $\BA\in\RR^{d\times m}$ consists of $m$ data points in $\RR^d$ and $\BW_i\in\RR^{d\times m}$ is the additive Gaussian noise. How to recover the underlying point cloud $\BA$ and corresponding orthogonal matrices $\{\BO_i\}_{i=1}^n$? This problem, known as the generalized orthogonal Procrustes problem, has found many applications in statistics~\cite{G75,SG02}, computer vision~\cite{BM92,MDKK16,MGPG04}, and imaging science~\cite{BKS14,CKS15,KK16,PSB21,S18,SS11}.

One common approach to solve this generalized orthogonal Procrustes problem (GOPP) is finding the least squares estimator. However, it turns out that finding the globally optimal least squares estimator is an NP-hard problem in general. In this work, will study the convex relaxation (semidefinite programming relaxation)~\cite{BKS14,BBS17,WS13} of this challenging nonconvex problem and also the performance of the generalized power method~\cite{B16,CC18,LYS17,LYS20,L20c,L21a,ZB18} to estimate the $\BA$ and $\BO_i$. 
In particular, we will focus on the following key questions:
\begin{center}
{\em When is the convex relaxation tight?}
\end{center}
Here the tightness means the solution to the convex relaxation exactly equals the least square estimator~\cite{BBS17,L20a,WS13,WLZS21,Z19}, also informally known as ``relax but no need to round". In other words, we are interested in understanding when a seemingly NP-hard problem can be solved by an algorithm of polynomial time complexity. This is motivated by an interesting numerical phenomenon: the tightness of convex relaxation holds if the noise level is small. In fact, it is conjectured there exists a threshold for the noise level below which the tightness of convex relaxation holds in~\cite{BKS14} by Bandeira, Khoo, and Singer as an open problem in 2014. In the author's earlier work~\cite{L21a}, it is shown that $\sigma\lesssim \sqrt{m}/(\sqrt{d} + \sqrt{m})$ suffices to ensure the tightness of the SDP (semidefinite program) relaxation. However, this theoretical bound does not match the current numerical simulations since the experiments indicate that the more plausible threshold should be $\sigma\lesssim \sqrt{n}/(\sqrt{nd} + \sqrt{m})$ (modulo a log factor). As a result, the first question this manuscript tries to answer is to give a more precise characterization of this threshold. 

As the SDP relaxation is usually computationally expensive, it is more favorable to use efficient first-order nonconvex methods such as the generalized power method. As the original least squares problem is highly nonconvex, we cannot always expect the generalized power method to find the least squares estimator.  Thus it is natural to ask
\begin{center}
{\em When can the generalized power method recover the global optimal solution?}
\end{center}
Numerical simulations indicate a similar phenomenon: if the noise is small, then the generalized power method indeed recovers the SDP solution, which is also the global optimal least squares estimator. Therefore, we are interested in understanding why the generalized power method with a proper initialization enjoys the global convergence even for this nonconvex problem. 
The core variable is the signal-to-noise ratio (SNR): it is widely believed that the stronger the noise, the harder the problem. We will study how the answers to these two aforementioned questions depend on the signal-to-noise ratio. We aim to find out the maximum level of noise which allows the tightness of the SDP relaxation and the global convergence of the generalized power methods~\cite{BBS17,LYS20,ZB18}.

\subsection{Related works}

The generalized orthogonal Procrustes problem (GOPP) has been studied under many different settings. For its broad applications, we refer the interested readers to~\cite{GD04,G75,SG02,BM92,MDKK16,MGPG04,BKS14,CKS15,KK16,PSB21,S18,SS11} and the reference therein. Our work will mainly focus on the optimization approaches in finding the least squares estimator of the GOPP. More precisely, finding the least squares estimator can be finally reduced to solving following optimization program with orthogonality constraints:
\begin{equation}\label{def:ls0}
\max_{\BO_i\in\Od(d)}~\sum_{i,j} \Tr(\BC_{ij}\BO_j\BO_i^{\top})
\end{equation}
where $\BC_{ij} = \BA_i\BA_j^{\top}\in\RR^{d\times d}$ is the ``cross-covariance" between the $i$th and $j$th point clouds, and 
\begin{equation}\label{def:od}
\Od(d) : = \{\BR\in\RR^{d\times d}: \BR\BR^{\top} =\BR^{\top}\BR=\I_d\}
\end{equation}
is the set of all $d\times d$ orthogonal matrices.
The program in this form~\eqref{def:ls0} (as well as its famous special case $d=1$) is a well-known NP-hard problem in general and has been widely studied~\cite{A17,BKS16,GW95,S11b,WY13}. Due its nonconvexity, a global optimal solution is often hard to obtain. Here we briefly review several popular approaches in dealing with~\eqref{def:ls0}.

Convex relaxation is arguably one of the most popular and powerful methods. One common way of relaxing~\eqref{def:ls0} is to generalize the famous Goemans-Williamson relaxation~\cite{GW95} for graph max-cut ($d=1$) to $d\geq 2$ by considering semidefinite program (SDP) with block diagonal constraints~\cite{BKS16,B15,NRV13}. After solving the relaxation, one would find an approximate solution via rounding with solid theoretical guarantees~\cite{BKS16,GW95,MMMO17}. Our work follows a different route by studying the tightness of the SDP relaxation: in many practical applications, we see that the SDP relaxation is sufficient to produce the global optimal solution to the original program. In other words, the rounding procedure is not necessary sometimes in examples including synchronization~\cite{BBS17,L20a,L20c,WS13,Z19,ZB18}, community detection~\cite{A17}, signal processing and machine learning~\cite{ARR13,CESV15,CR09}. However, the tightness is not free lunch as it usually holds when the noise in the input data is relatively small. Therefore, our goal is to identify the maximum noise level for the tightness to hold. 

On the other hand, despite the success of convex relaxation and the recent progresses of general SDP solvers~\cite{STYZ20,TTT03,YST15}, solving large-scale SDPs could be still highly expensive and even computationally prohibitive in practice. Therefore, efficient first-order iterative methods are often preferred~\cite{AMS09,BM03,BM05,BVB20,YTFU21,WY13}. In the past few years, there are significant progresses in developing nonconvex approaches  with theoretical guarantees for various problems  in signal processing and machine learning. The recent advances in the provably convergent nonconvex methods roughly fall into two categories: (a) Showing that the optimization landscape is benign, i.e., the objective function has only one local optimum (which is also global) and there are no other ``spurious" local optima~\cite{BM03,BM05,BVB20,GLM16,L20a,MMMO17,SQW18}; (b) Construct a smart initialization for the algorithm and then show that the first-order method produces a sequence of iterates which converge to the global optimum~\cite{CLS15,CC18,KMO10,LYS17,LYS20,MWCC20}. Our work follows the second idea: we will first perform spectral initialization and then show that the first-order method will enjoy linear global convergence to the global optima. In particular, we will focus on the studying the convergence of the generalized power method (GPM). The GPM has been successfully applied to several related problems, especially group synchronization~\cite{B16,CC18,L20c,LYS17,LYS20,ZB18}: once the initialization is sufficiently close to the global optima, one can obtain the global convergence of the GPM. 
Our work benefits greatly from the recent popular leave-one-out technique~\cite{MWCC20,ZB18} in analyzing the convergence of our algorithm. In particular, the first step of our algorithm is the spectral initialization which needs a block-wise error bound on the singular vectors of a spiked rectangle matrix. This block-wise error bound is made possibly by the leave-one-out technique in obtaining an $\ell_{\infty}$-norm and block-wise error bound of the eigenvector perturbation~\cite{AFWZ20,L20b,ZB18}. 

The generalized Procrustes problem is closely related to the orthogonal group synchronization problem~\cite{L20c,LYS20,MMMO17,S11,WS13,ZWS21}. The $\Od(d)$ synchronization asks to recover a set of group elements $\BO_i\in\Od(d)$ from their noisy pairwise measurements:
\begin{equation}\label{def:odgroup}
\BO_{ij} : = \BO_i \BO_j^{-1} + \BE_{ij},\quad (i,j)\in{\cal E}
\end{equation}
where each $\BO_i$ is orthogonal, $\BE_{ij}$ is additive noise (it also can be multiplicative noise), and ${\cal E}$ is the edge set of all observed pairwise measurements. One benchmark statistical model for $\Od(d)$ synchronization is: each noise $\BE_{ij}$ is an independent random matrix (such as Gaussian random matrix) for all $i<j$. The GOPP is similar to the group synchronization in the sense that the pairwise measurement $\BC_{ij} = \BA_i\BA_j^{\top}$ contains the information about $\BO_i$ and $\BO_j$. However, the most obvious difference is that the noise between two point clouds $\BA_i$ and $\BA_jj$ is not independent of the edge unlike the benchmark group synchronization model~\eqref{def:odgroup} with additive independent noise on edges~\cite{A17,ABBS14,BBS17,WS13,ZB18}. This difference in the setup results in significant technical difference in the convergence analysis of the generalized power method from the previous works~\cite{B16,CC18,L20c,LYS17,LYS20,ZB18}. 

We conclude this section by summarizing the main contribution of our work. By using the leave-one-out technique, we successfully analyze the convergence of the generalized power method in the generalized orthogonal Procrustes problem. 
In particular, we show that the generalized power method converges linearly to a limiting point which is equal to the least squares estimator and also corresponds to the optimal solution to the SDP relaxation, if the noise strength $\sigma$ satisfies $\sigma \lesssim \sqrt{n}/(\sqrt{d}(\sqrt{nd} + \sqrt{m} + 2\sqrt{n\log n}))$ (modulo the constants, log-factor and condition number). As a result, our result provides a near-optimal sufficient condition for the tightness of the SDP relaxation in terms of information-theoretic bound and resolves the open problem proposed by Bandeira, Khoo, and Singer in~\cite{BKS14}. 

\subsection{Organization of our paper} 
Our paper will proceed as follows. In Section~\ref{s:prelim}, we will introduce the basics of the generalized orthogonal Procrustes problem and several optimization formulations including semidefinite relaxation and generalized power method. Section~\ref{s:mainthm} will focus on the main theoretical results in this paper and Section~\ref{s:numerics} will present several numerical experiments to complement our theory.  The proof will be presented in Section~\ref{s:proof}.

\subsection{Notations}
We denote vectors and matrices by boldface letters $\bx$ and $\BX$ respectively. For a given matrix $\BX$, $\BX^{\top}$ is the transpose of $\BX$ and $\BX\succeq 0$ means $\BX$ is symmetric and positive semidefinite. Let $\I_n$ be the identity matrix of size $n\times n$. For two matrices  $\BX$ and $\BY$ of the same size, their inner product is $\lag \BX,\BY\rag:= \Tr(\BX^{\top}\BY) = \sum_{i,j}X_{ij}Y_{ij}.$ Let $\|\BX\|$ be the operator norm, $\|\BX\|_*$ be the nuclear norm, and $\|\BX\|_F$ be the Frobenius norm. We denote the $i$th largest and the smallest singular value of $\BX$ by $\sigma_{i}(\BX)$ and $\sigma_{\min}(\BX)$ respectively.
We let $\Od(d) : = \{\BO\in\RR^{d\times d}: \BR^{\top}\BR= \BR\BR^{\top}=\I_d\}$ be the set of all $d\times d$ orthogonal matrices.
For a non-negative function $f(x)$, we write $f(x)\lesssim g(x)$ and $f(x) = O(g(x))$ if there exists a positive constant $C_0$ such that $f(x)\leq C_0g(x)$ for all $x.$

\section{Preliminaries}\label{s:prelim}

Suppose we observe a set of $n$ noisy point clouds in $\RR^d$ transformed by an unknown orthogonal matrix $\BO_i\in\Od(d)$:
\begin{equation}\label{def:model}
\BA_i = \BO_i\BA + \sigma\BW_i
\end{equation}
where $\BW_i$ is the additive Gaussian noise and $m \geq d$. Without loss of generality, we assume the underlying orthogonal matrix $\BO_i$ is $\I_d$ for all $1\leq i\leq n.$
Then the model becomes
\[
\BA_i = \BA + \sigma\BW_i, \quad \BA\in\RR^{d\times m}
\]
and we can write the observed data into the signal-plus-noise form:
\begin{equation}\label{def:Z}
\BD := \BZ\BA + \sigma\BW\in\RR^{nd\times m}, \quad \BZ^{\top}:= [\I_d,\cdots,\I_d] \in\RR^{d\times nd} 
\end{equation}
where $\BW\in\RR^{nd\times m}$ is a Gaussian random matrix.

The key question is: how to recover $\BO_i$ and $\BA$ from $\BA_i$ (equivalently $\BD$)?  One commonly-used approach is finding the least squares estimator of $\BO_i$ and $\BA$ by minimizing \begin{equation}\label{def:ls}
\min_{\BO_i\in\Od(d),\BA\in\RR^{d\times m}}~ \sum_{i=1}^n \|\BA_i - \BO_i\BA\|_F^2.
\end{equation}
Under the additive Gaussian noise, the least squares estimator exactly matches the maximum likelihood estimator. It is easy to see that~\eqref{def:ls} is a nonconvex program due to the orthogonality constraints. Note that~\eqref{def:ls} involves two variables $\BO_i$ and $\BA$. In fact, we can remove $\BA$ by first fixing $\BO_i$.
By assuming $\BO_i$ is known,~\eqref{def:ls} is a convex function over $\BA$ and minimizing it gives 
\begin{equation}\label{eq:Ahat}
\widehat{\BA} = \frac{1}{n}\sum_{i=1}^n \BO_i^{\top}\BA_i.
\end{equation}
By substituting $\widehat{\BA}$ back into~\eqref{def:ls}, the objective function~\eqref{def:ls} is equivalent to maximizing a generalized quadratic form with orthogonality constraints
\begin{equation}\label{def:od}
\max_{\BO_i\in \Od(d)}~\sum_{i<j}\lag \BC_{ij}, \BO_i\BO_j^{\top}\rag, \quad\text{where}\quad\BC_{ij} = \BA_i\BA_j^{\top} \tag{P}
\end{equation}
which gives~\eqref{def:ls0}. An equivalent formulation is given by
\begin{equation}\label{def:fo}
f(\BO) = \left\lag \BC, \BO\BO^{\top}\right\rag = \| \BD^{\top}\BO \|_F^2, \qquad \BC := \BD\BD^{\top}\in\RR^{nd\times nd}
\end{equation}
where 
\begin{equation}\label{def:D}
\BO^{\top} = [\BO_1^{\top}, \cdots, \BO_n^{\top}]\in\Od(d)^{\otimes n}, \qquad \BD^{\top} = [\BA_1^{\top}, \cdots, \BA_n^{\top}]\in\RR^{m\times nd}.
\end{equation}
From now on, we will focus on recovering the orthogonal matrices $\BO_i$ by solving~\eqref{def:od} and then use the estimator of $\BO_i$ to reconstruct $\BA$ via~\eqref{eq:Ahat}.

\subsection{Semidefinite programming relaxation}
The nonconvex program~\eqref{def:od} is a famous NP-hard problem in general. One approach to tackle it is to use the semidefinite program (SDP) relaxation by using the fact that
\[
\BO\BO^{\top}\succeq 0, \qquad [\BO\BO^{\top}]_{ii} = \BO_i\BO_i^{\top} = \I_d.
\]
As a result, the SDP relaxation of~\eqref{def:od} is 
\begin{equation}\label{def:sdp}
\max~\lag \BC,\BG\rag\quad \text{such that} \quad \BG\succeq 0, ~\BG_{ii} = \I_d,\tag{SDP}
\end{equation}
where $\BO\BO^{\top}$ is replaced by $\BG  \in\RR^{nd\times nd}.$ It is a natural generalization of the Goemans-Williamson relaxation $d=1$ of the graph max-cut problem~\cite{GW95} to $d\geq 2$. This SDP with block-diagonal constraints has been found in several applications include group synchronization~\cite{B15,WS13,RCBL19}.
Many off-the-shelf SDP solvers~\cite{STYZ20,TTT03,YST15} can find the global optimal solution to~\eqref{def:sdp}  directly for a relatively small size program. 
In general, one cannot expect the solution to the relaxation~\eqref{def:sdp} is equal to that to~\eqref{def:od}. However, numerical experiments in~\cite{L21a} indicate that the~\eqref{def:sdp} relaxation yields the global optimal solution to~\eqref{def:od} for $\sigma\lesssim \sqrt{n}/(\sqrt{nd} + \sqrt{m} + 2\sqrt{n\log n})$, i.e., when the noise level is small compared with the signal strength. Therefore, one core component of our work is to justify the tightness of~\eqref{def:sdp} under the high SNR (signal-to-noise ratio) regime, i.e., for relatively small $\sigma.$

\subsection{Generalized power method and alternating minimization algorithm}
As discussed briefly in the introduction, it is computationally prohibitive to solve a large-scale SDP directly, which is the major roadblock towards practical applications. Thus we aim to develop efficient first-order methods which achieve similar performance as the SDP does while also enjoying solid theoretical guarantees. Inspired by the recent progresses of the generalized power method (GPM) in group synchronization~\cite{B16,LYS17,LYS20,L20c}, we apply the GPM to the generalized orthogonal Procrustes problem. 

The generalized power method can be also viewed as the conditional gradient method with a specific stepsize~\cite{LYS17}. Note that by dropping the orthogonality constraints in~\eqref{def:od}, the gradient of $f(\BO)$ in~\eqref{def:fo} is
\[
\nabla f(\BO) = 2\BC\BO.
\]
Given the current iterate $\BS^t\in\Od(d)^{\otimes n}$, we update the next iterate by
\[
\BS^{t+1} : = \argmax_{\BO\in\Od(d)^{\otimes n}} \lag \BO, \BC\BS^t\rag = \sum_{i=1}^n \argmax_{\BO_i\in\Od(d)}  \lag \BO_i, [\BC\BS^t]_i\rag.
\]
We can view this updating scheme as one-step power iteration followed by projection step. In fact, it is not hard to find the explicit form of $\BS^{t+1}$ by using the singular value decomposition: the $i$th $d\times d$ block $\BS^{t+1}_i$ of $\BS^{t+1}$ is given by
\[
\BS^{t+1}_i = \PP([\BC\BS^t]_i),~~1\leq i\leq n,
\]
where
\begin{equation}\label{def:P}
\PP(\BX) := \BU\BV^{\top}
\end{equation}
is also known as the matrix sign function~\cite{G11}
where $\BU$ and $\BV\in\Od(d)$ are the normalized sleft/right singular vectors of $\BX$. Actually, $\PP(\BX)$ is one global minimizer (not necessarily unique if $\BX$ is not invertible) to 
\[
\PP(\BX) : = \argmin_{\BQ\in\Od(d)}\|\BX - \BQ\|_F^2.
\]
By the definition of $\PP(\cdot)$, it has multiple choices of $\BU$ and $\BV$ if the input matrix is rank deficient. In this case, it suffices to pick one specific choice of left/right singular vectors. In particular, if $\BX$ is invertible, then
\[
\PP(\BX) = (\BX\BX^{\top})^{-1/2}\BX = \BX(\BX^{\top}\BX)^{-1/2}.
\]

In summary, the GPM aims to recover the hidden orthogonal matrices iteratively by using 
\begin{equation}\label{def:gpm}
\BS^{t+1} = \PP_n(\BC\BS^t) \tag{GPM}
\end{equation}
where $\PP_n(\cdot)$ is an operator from $\RR^{nd\times d}$ to $\Od(d)^{\otimes n}$, i.e.,
\begin{equation}\label{def:Pn}
\PP_n(\BX) := 
\begin{bmatrix}
\PP(\BX_1) \\
\vdots \\
\PP(\BX_n)
\end{bmatrix}
\end{equation}
which applies ${\cal P}$ operation to each $d\times d$ block, where $\BX_i$ is the $i$th $d\times d$ block of $\BX\in\RR^{nd\times d}.$

\vskip0.25cm
In fact,~\eqref{def:gpm} is equivalent to the alternating minimization algorithm applied to~\eqref{def:ls}. Suppose we are given $\BS^{t}\in\Od(d)^{\otimes n}$, then the next updated estimator of $\BA$ follows from~\eqref{eq:Ahat}:
\[
\BA^{t} = \frac{1}{n}\sum_{i=1}^n (\BS_i^{t})^{\top}\BA_i = \frac{1}{n}(\BS^{t})^{\top}\BD.
\]
Given $\BA^{t}$, we update the orthogonal matrix $\BO_i$ by solving $\min_{\BO_i\in\Od(d)}\|\BA_i - \BO_i\BA^{t} \|_F$ which gives
\[
\BS_i^{t+1} = \PP( \BA_i(\BA^{t})^{\top}  ) = \PP( \BA_i\BD^{\top}\BS^{t}  ) = \PP(\BD_i \BD^{\top} \BS^t) = \PP([\BC \BS^t]_i)
\]
where the $i$-th block of $\BD$ equals $\BA_i.$ 

\vskip0.25cm

Note that the GPM is essentially nonconvex and thus a carefully chosen initialization would help the convergence. 
Here we initialize the algorithm by using spectral method: first compute the top $d$ left singular vectors $\BU$ of $\BD$ in~\eqref{def:D} and then round each $d\times d$ block to an orthogonal matrix, 
\begin{equation}\label{def:init}
\BS^0 = \PP_n(\BU).
\end{equation}
As~\eqref{def:od} is a nonconvex program, we cannot hope that the global convergence always holds as there is still the risk of getting stuck at local optima or at other fixed points. 
Remarkably, several previous numerical experiments in~\cite{L21a} indicates that the generalized power method with spectral initialization would converge to the global optimal solution, provided that the noise level is relatively small. As a result, the major theoretical question is obtaining the convergence of the generalized power method with spectral initialization. 

To discuss the convergence of the generalized power method, we need to introduce a distance on $\RR^{nd\times d}$ which is invariant to a global orthogonal transform since the solutions to~\eqref{def:od} are equivalent if they are multiplied right by a common orthogonal matrix. To resolve this ambiguity, we define a distance between two matrices in $\RR^{nd\times d}$ by
\begin{equation}\label{def:df}
d_F(\BX,\BY) : = \min_{\BQ\in\Od(d)}\|\BX - \BY\BQ\|_F.
\end{equation}
The global minimizer is given by $\BQ = \PP(\BY^{\top}\BX)$ and
\begin{align*}
d_F(\BX,\BY) & = \sqrt{\|\BX\|_F^2 + \|\BY\|_F^2-2\lag \BX,\BY\BQ\rag}  \\
& =  \sqrt{\|\BX\|_F^2 + \|\BY\|_F^2-2\|\BY^{\top}\BX\|_*}
\end{align*}
where $\|\cdot\|_*$ denotes the nuclear norm and $\PP(\cdot)$ is defined in~\eqref{def:P}.
In particular, if we let $\BX$ and $\BY$ be in $\Od(d)^{\otimes n}$, then
\[
d_F(\BX,\BY) =  \sqrt{2nd-2\|\BY^{\top}\BX\|_*}
\]
since $\|\BX\|_F = \sqrt{nd}$ holds for $\BX\in\Od(d)^{\otimes n}.$

\section{Main results}\label{s:mainthm}

Now we are ready to introduce our first main result. 
The proof of the global convergence of the~\eqref{def:gpm} relies on spectral initialization, i.e., computing the top $d$ left singular vectors $\BU$ of $\BD$ in~\eqref{def:D} and then round each $d\times d$ block $\BU_i$ of $\BU$ to an orthogonal matrix. As a by-product, we first provide a block-wise error bound for the spectral estimation of $\BO$ in~\eqref{def:D}.
\begin{theorem}\label{thm:spectral}
Suppose we observe $\BA_i = \BO_i \BA + \sigma\BW_i$, $1\leq i\leq n$, with $\BA\in\RR^{d\times m}$ $(m\geq d)$ and additive Gaussian noise $\BW_i$. 
The spectral estimator $\BS^0\in\Od(d)^{\otimes n}$ in~\eqref{def:init} of $\BO\in\Od(d)^{\otimes n}$ satisfies
\begin{equation}
\min_{\BQ\in\Od(d)}\max_{1\leq i\leq n} \| \BS_i^0 - \BO_i\BQ \| \lesssim \frac{\sigma\kappa^2}{\sigma_{\min}(\BA)}\cdot \frac{(\sqrt{nd} + \sqrt{m}+\sqrt{2\gamma n\log n})}{\sqrt{n}}
\end{equation}
with probability at least $1-O(n^{-\gamma+2})$ if
\[
\sigma \lesssim  \frac{\sigma_{\min}(\BA)}{\kappa^4}\cdot \frac{\sqrt{n}}{\sqrt{nd} + \sqrt{m} + \sqrt{2\gamma n\log n}}, \quad \kappa := \frac{\sigma_{\max}(\BA)}{\sigma_{\min}(\BA)}.
\]
\end{theorem}
Here we implicitly require the smallest singular value $\sigma_{\min}(\BA)$ of $\BA$ is strictly positive, i.e., $\BA$ is rank-$d$ and $m\geq d$.
This result is a natural generalization of the $\ell_{\infty}$-norm singular vector perturbation bound to the block-wise error bound~\cite{AFWZ20}. Note that applying Davis-Kahan~\cite{DK70} or Wedin~\cite{W72} type perturbation argument for singular vectors gives an $\ell_2$-norm bound:
\[
\min_{\BQ\in\Od(d)}\|\BU - \BO\BQ\| \lesssim \frac{\sigma\sqrt{n}\|\BW \|}{\sigma_{d}(\BD)} \lesssim \frac{\sigma(\sqrt{nd} + \sqrt{m})}{\sigma_{\min}(\BA)}
\]
where we choose to normalize the top $d$ left singular vectors $\BU\in\RR^{nd\times d}$ of $\BD$ as $\BU^{\top}\BU=n\I_d$ and $\sigma_d(\BD)$ is the $d$th largest singular value of $\BD$ which satisfies
\[
\sigma_d(\BD) \geq \sigma_{d}(\BZ\BA) - \sigma\|\BW\| \geq \sqrt{n}\sigma_{\min}(\BA) - \sigma (\sqrt{nd} + \sqrt{m}) \gtrsim \sqrt{n}\sigma_{\min}(\BA)
\]
under 
\[
\sigma \lesssim \frac{\sqrt{n}\sigma_{\min}(\BA)}{\sqrt{nd}+\sqrt{m}}.
\]
Theorem~\ref{thm:spectral} implies that the error incurred by the additive Gaussian random noise is evenly distributed to each block. Also it means the spectral method provides an initialization $\BS^0$ which is sufficiently close to the ground truth signal $\BO$ blockwisely. 

\vskip0.25cm

Note that the spectral initialization in general does not yield the exact maximum likelihood estimator to~\eqref{def:od}. Thus we will start from this initial guess~\eqref{def:init} and apply the generalized power method~\eqref{def:gpm}. 
Our second main result is the global linear convergence of the generalized power method with spectral initialization to the least squares estimator (also the maximum likelihood estimator) of the orthogonal matrices $\BO_i$. In fact, the tightness of the SDP relaxation is a by-product of the global convergence of the GPM.

\begin{theorem}\label{thm:main}
Suppose we observe $\BA_i = \BO_i \BA + \sigma\BW_i$, $1\leq i\leq n$, with $\BA\in\RR^{d\times m}$ and additive Gaussian noise $\BW_i$. With probability at least $1-O(n^{-\gamma+2})$, the sequence $\{\BS^t\}_{t\geq 0}$ from~\eqref{def:gpm} converges to a unique fixed point $\BS^{\infty}$ linearly 
\[
d_F(\BS^t,\BS^{\infty}) \leq \rho^t d_F(\BS^0,\BS^{\infty})
\]
where $\rho<1$ is a constant and the initialization is $\BS^0 = \PP_n(\BU)$ with the top $d$ left singular vectors $\BU\in\RR^{nd\times d}$ of $\BD\in\RR^{nd\times m}$, provided that
\begin{equation}\label{cond:sigma}
\sigma \lesssim  \frac{\sigma_{\min}(\BA)}{\kappa^4}\cdot \frac{\sqrt{n}}{\sqrt{d}(\sqrt{nd} + \sqrt{m} + \sqrt{2\gamma n\log n})}, \quad \kappa = \frac{\sigma_{\max}(\BA)}{\sigma_{\min}(\BA)}.
\end{equation}
Moreover, $\BS^{\infty}$ is the unique global maximizer to~\eqref{def:od} (modulo a global orthogonal transform) and $\BS^{\infty}(\BS^{\infty})^{\top}$ is the unique global maximizer to the~\eqref{def:sdp} relaxation.
\end{theorem}

Here are a few remarks on the Theorem~\ref{thm:main}. First, it provides a sufficient condition on the tightness of the SDP relaxation: solving~\eqref{def:sdp} suffices to produce the global optimal solution to the original program~\eqref{def:od}. Moreover, one does not need to solve the SDP directly using a standard SDP solver; instead, using efficient generalized power method would give the exact global optimal solution. In fact, the SDP relaxation and generalized power method have quite similar empirical performance, as shown in our numerical experiments in Section~\ref{s:numerics}.

The main difference of this work from the author's earlier work~\cite{L21a} is the improvement on the bound of $\sigma$. In~\cite{L21a}, it has been shown that $\sigma \lesssim \sigma_{\min}(\BA)/(\kappa^4\sqrt{d}(\sqrt{d} + \sqrt{m} + \sqrt{2 \log n}))$ is needed to ensure the tightness\footnote{While the result~\cite{L21a} holds for more general noise, it is suboptimal when applied to additive Gaussian noise.}: this bound essentially requires the additive noise in each point cloud is smaller than its signal strength. Therefore, the bound on $\sigma$ is nearly independent of the number of point clouds $n$. Instead, the new bound~\eqref{cond:sigma} integrates the information of all pairwise point clouds: as the sample size $n$ gets larger, the right hand side of~\eqref{cond:sigma}, i.e., maximum allowable noise level, increases
\[
\sigma \lesssim  \frac{\sigma_{\min}(\BA)}{\kappa^4}\cdot \frac{1}{\sqrt{d}(\sqrt{d} + \sqrt{m/n} + \sqrt{2\gamma\log n})}
\]
as $m/n$ decreases, which means the tightness of the SDP relaxation and the global convergence of the~\eqref{def:gpm} hold in the presence of stronger noise than the bound given in~\cite{L21a}.

One natural question is about the optimality of~\eqref{cond:sigma}. It has been shown in many applications that if the signal strength is stronger than the noise, it is possible to recover a meaningful solution~\cite{A17,AFWZ20,B16,CC18,S11,ZB18}. In our case where the data matrix $\BD = \BZ\BA + \sigma\BW$ is in the signal-plus-noise form, we would expect that the tightness of the SDP relaxation and  the convergence of the GPM hold
\[
 \frac{\sqrt{n}\|\BA\|}{\sigma\|\BW\|} > 1 \Longleftrightarrow \sigma < \frac{\sqrt{n}\|\BA\|}{\|\BW\|}.
\]
Note that $\|\BW\| \lesssim \sqrt{nd} + \sqrt{m}$ and thus we hope that 
\[
\sigma \lesssim \frac{\sqrt{n}\|\BA\|}{\sqrt{nd} + \sqrt{m}}
\]
is needed to yield a nontrivial recovery of $\{\BO_i\}_{i=1}^n$ via the GPM with spectral initialization. In fact, if we assume $\BA$ is partially orthogonal, i.e., $\BA\BA^{\top} = \I_d$, then it is believed that the information-theoretic bound for the detection these signals should be $\sigma \approx \sqrt{n}/(\sqrt{nd} + \sqrt{m})$~\cite{BN11,JCL20,OVW18,PWBM18}. 
Our result differs from this near-optimal bound by a factor of $\sqrt{d}$,  condition number $\kappa$, and a log term. Numerical simulations in Section~\ref{s:numerics} indicate that the presence of $\sqrt{d}$ and condition number $\kappa$ in~\eqref{cond:sigma} do not seem necessary. Hence, we conjecture that the tightness of the SDP relaxation would hold if
\[
\sigma\lesssim \frac{\sqrt{n} \sigma_{\min}(\BA) }{\sqrt{nd} +\sqrt{m} + c\sqrt{\log n}}
\]
for some constant $c$ with high probability.

The next main theorem characterizes the estimation blockwise error bound of least squares estimator $\BS_i^{\infty}$ to the true planted orthogonal matrix $\BO_i$, and also the estimation error $\widehat{\BA}$ of the hidden point  cloud $\BA$ up to a global orthogonal transform.

\begin{theorem}\label{thm:error}
The limiting point $\BS^{\infty}$ from the~\eqref{def:gpm} satisfies
\begin{align*}
\min_{\BQ\in\Od(d)}\max_{1\leq i\leq n} \| \BS_i^{\infty} - \BQ \|_F & \lesssim \frac{\kappa^2}{\|\BA\|}  \cdot \frac{\sigma \sqrt{d}(\sqrt{nd} + \sqrt{m} + \sqrt{2\gamma n\log n})}{\sqrt{n}}, \\
\min_{\BQ\in\Od(d)}\left\| \frac{1}{n}\sum_{i=1}^n\BS_i^{\infty}  \BA_i  - \BQ\BA\right\|_F &\lesssim \frac{\kappa^2\sigma \sqrt{d}(\sqrt{nd} + \sqrt{m} + \sqrt{2\gamma n\log n})}{\sqrt{n}},
\end{align*}
with probability at least $1-O(n^{-\gamma+2})$
provided that
\[
\sigma \lesssim  \frac{\sigma_{\min}(\BA)}{\kappa^4}\cdot \frac{\sqrt{n}}{\sqrt{d}(\sqrt{nd} + \sqrt{m}+\sqrt{2\gamma n\log n})}.
\]
\end{theorem}
Theorem~\ref{thm:error} indicates that the estimation error of each orthogonal matrix is approximately of the same order. 

We conclude this section with a few possible future directions. One major problem is the possible sub-optimality of~\eqref{cond:sigma} which differs from the information-theoretic limit by a factor of $\sqrt{d}$. Similar to many other low-rank matrix recovery problems~\cite{L20c,MWCC20} via nonconvex approaches, the dependence on the rank of the planted signal (Here the rank equals $d$) is typically suboptimal. To overcome this issue, one may need to come up with a sharper convergence analysis. In our case, it would be the characterization of the contraction mapping induced by the generalized power method under operator norm instead of Frobenius norm~\eqref{def:df}. In fact, in our work, the major technical bottleneck arises in Lemma~\ref{lem:Pncon}: it is unclear to how to establish the Lipschitz continuity for the operator $\PP_n(\cdot)$ under operator norm $d_2(\BX,\BY): = \min_{\BQ\in\Od(d)}\|\BX-\BY\BQ\|$ instead of $d_F(\cdot,\cdot)$ in~\eqref{def:df}.
Another interesting problem would be how to find the exact critical threshold for the tightness of the SDP relaxation as our current analysis relies on the GPM to understand the output of the convex relaxation. However, it is very likely that there is still a performance gap (despite this gap could be small, as shown in Section~\ref{s:numerics}) between the GPM and convex relaxation, which cannot be fully characterized by our current methods. The similar issues also can be found in phase synchronization~\cite{BBS17,ZB18} and orthogonal group synchronization~\cite{L20c} since the ground truth rotation is not usually the global optimal solution to the SDP relaxation. We leave these problems to the future works.

\section{Numerics}\label{s:numerics}

This section is devoted to a few numerical experiments which complement our theoretical results. 

\subsection{Performance comparison between~\eqref{def:sdp} relaxation and~\eqref{def:gpm} }

Our theory on the tightness of the~\eqref{def:sdp}  relaxation is established via the global convergence of the generalized power method~\eqref{def:gpm}. Therefore, it is natural to ask if there is a performance gap between these two different approaches in obtaining the global optimal solution to~\eqref{def:od}. To answer this question, we set the underlying data matrix $\BA$ as $\BA\BA^{\top} = \I_d$ and $\BO_i=\I_d$ in~\eqref{def:model}, and rescale the noise strength $\sigma$ as
\[
\sigma = \frac{\eta\sqrt{n}}{\sqrt{nd}+\sqrt{m}}, \quad \eta>0.
\]
The noise parameter $\eta$ ranges from 0 to 1.2, and the dimension factor $(m,n,d)$ is fixed. For the tightness of~\eqref{def:sdp} relaxation, we use the SDPNAL+ solver~\cite{YST15,STYZ20} to obtain the~\eqref{def:sdp} solution exactly and then check if the resulting SDP solution is exactly rank $d$, i.e., the SDP solution produces the global solution to~\eqref{def:od}. For the convergence of the~\eqref{def:gpm}, we simply initialize the algorithm with spectral method~\eqref{def:init}, followed by the generalized power iteration. Then once the iterates stabilize or reach the maximum number of iterations, we certify the global optimality of the~\eqref{def:gpm} solution via the duality theory in convex optimization, more precisely, in Theorem~\ref{thm:cvx}. 

\begin{figure}[h!]
\centering
\begin{minipage}{0.48\textwidth}
\includegraphics[width=80mm]{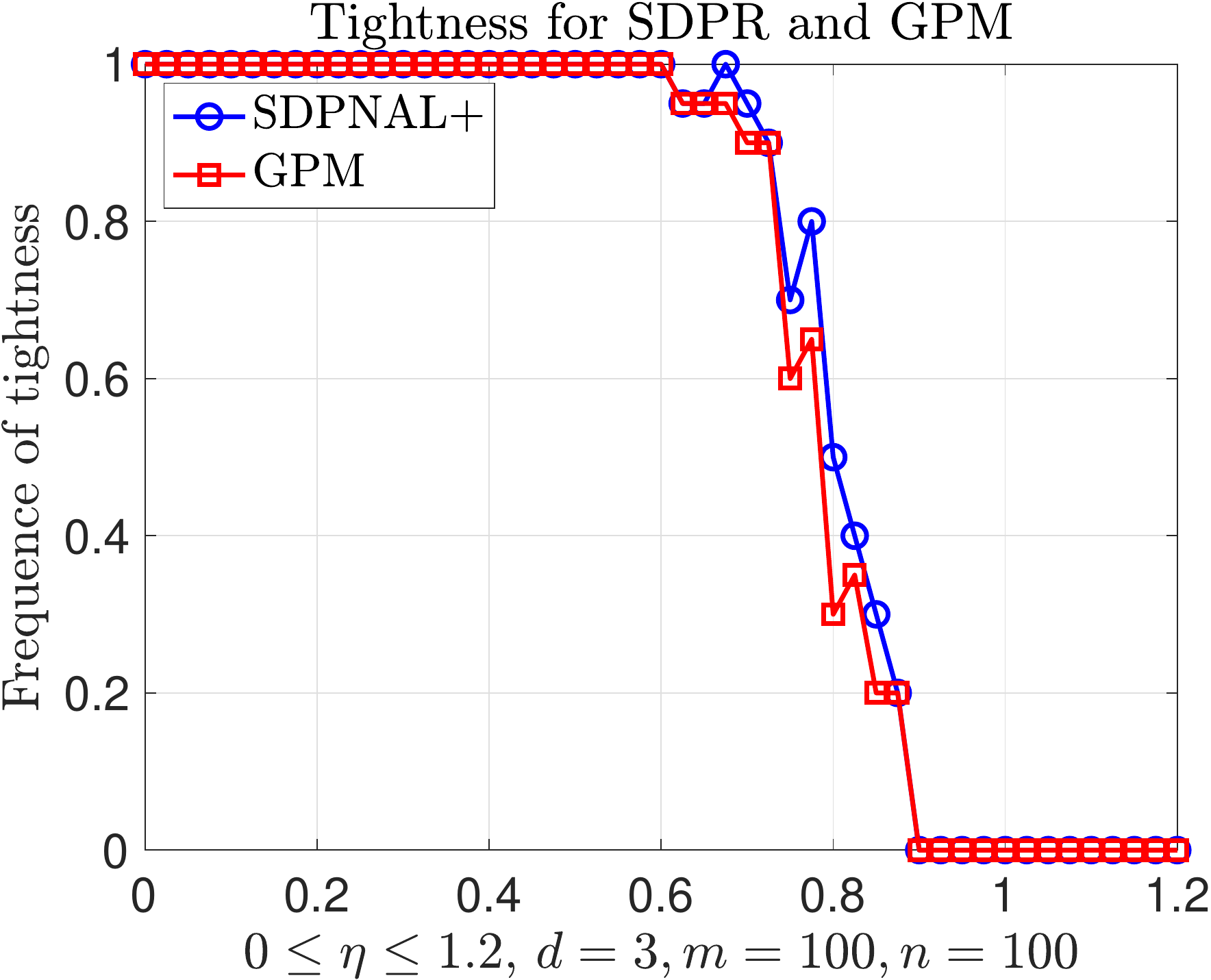}
\end{minipage}
\hfill
\begin{minipage}{0.48\textwidth}
\includegraphics[width=80mm]{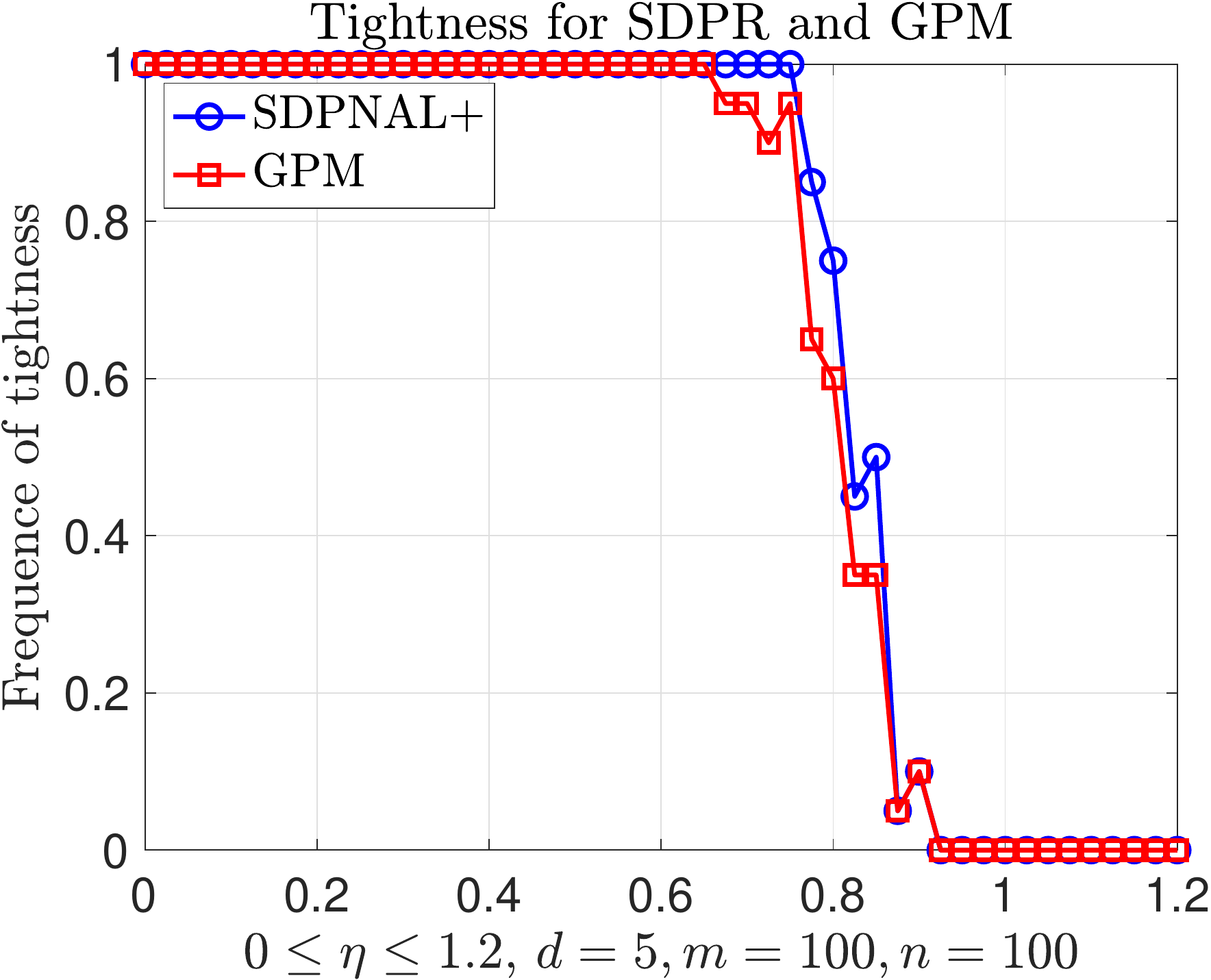}
\end{minipage}
\caption{Performance comparison between the SDP relaxation and the GPM: the parameters $(d,n,m)$ are fixed; for each $\eta$, we run 20 experiments and compute the frequency of the tightness of the SDPR (blue curve) and the global convergence of the GPM (red curve). }
\label{fig:1}
\end{figure}

Figure~\ref{fig:1} gives the comparison between the SDP relaxation (SDPR) and the GPM: both the SDPR and GPM successfully produce the tight SDP solution for $\eta<0.6$, and then the probability of tightness rapidly decays to 0 as $\eta$ increases from 0.6 to 0.9. In conclusion, we can see that the SDPR outperforms the GPM marginally but their overall performances are highly similar.

\subsection{Global convergence of the~\eqref{def:gpm}}

Our next experiment aims to verify if our bound~\eqref{cond:sigma} on $\sigma$ in Theorem~\ref{thm:main} has the optimal scaling, i.e., $\sigma\lesssim \sqrt{n}/(\sqrt{d} (\sqrt{nd}+ \sqrt{m} + \sqrt{2\gamma n \log n}))$ with probability at least $1-O(n^{-\gamma+1}).$ We set different $(d,n,m)$ in the experiment, and the noise level $\eta$ varies from 0 to 1.2. For each set of $(d,n,m,\eta)$, we run 20 experiments and calculate the proportion of the global convergence of the~\eqref{def:gpm}. 

\begin{figure}[h!]
\centering
\begin{minipage}{0.48\textwidth}
\includegraphics[width=80mm]{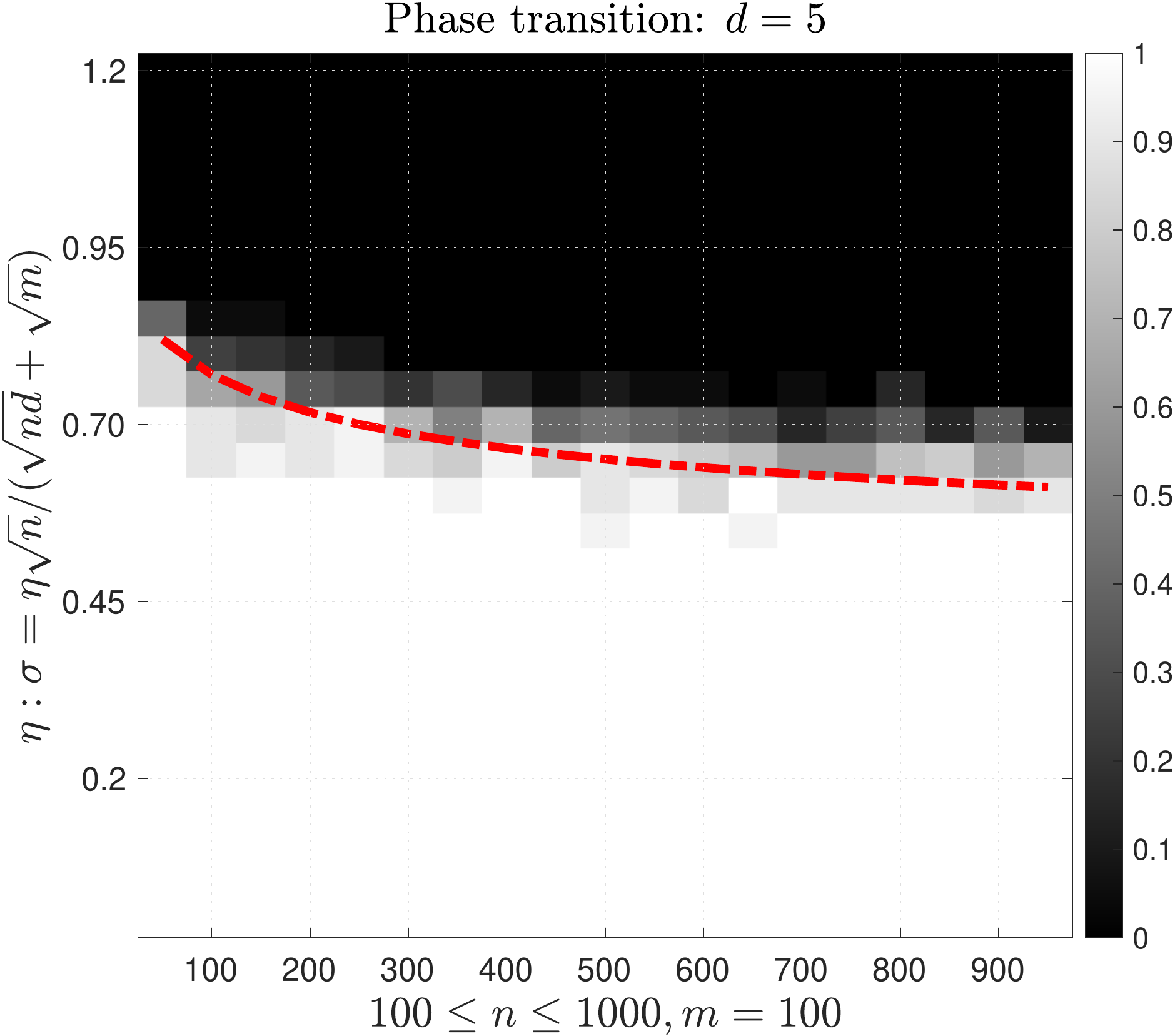}
\end{minipage}
\hfill
\begin{minipage}{0.48\textwidth}
\includegraphics[width=80mm]{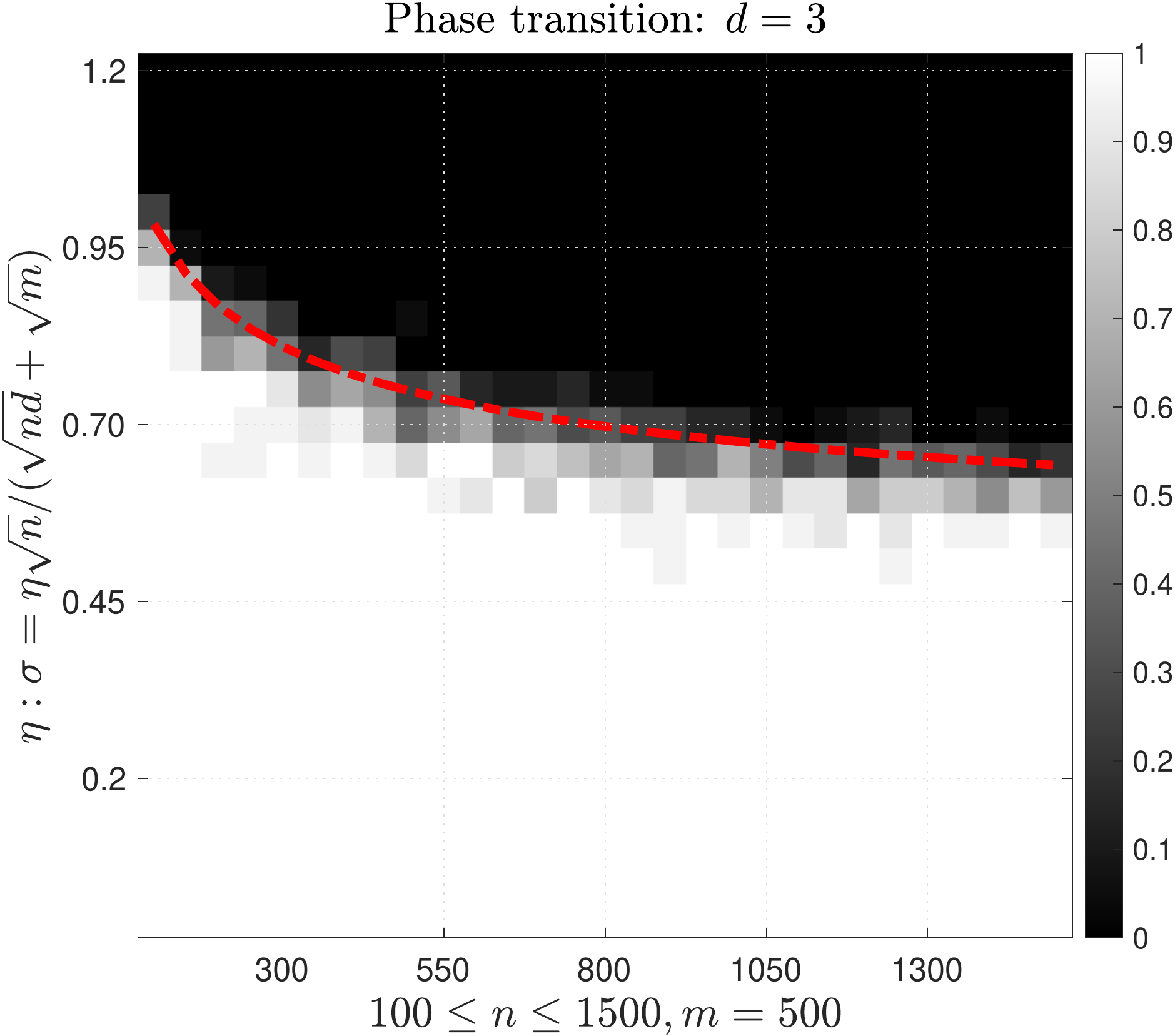}
\end{minipage}
\vfill
\vskip0.2cm
\begin{minipage}{0.48\textwidth}
\includegraphics[width=80mm]{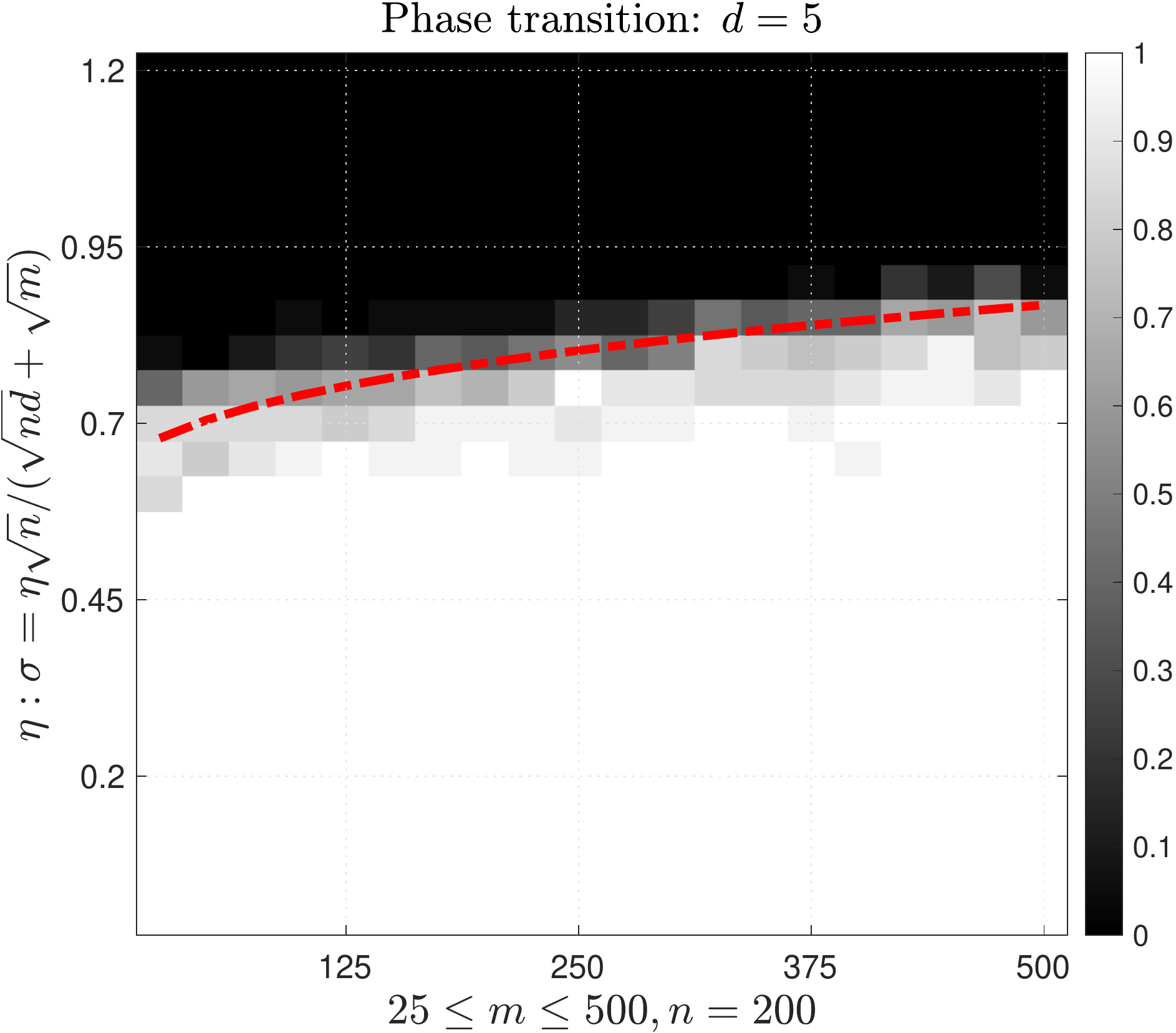}
\end{minipage}
\hfill
\begin{minipage}{0.48\textwidth}
\includegraphics[width=80mm]{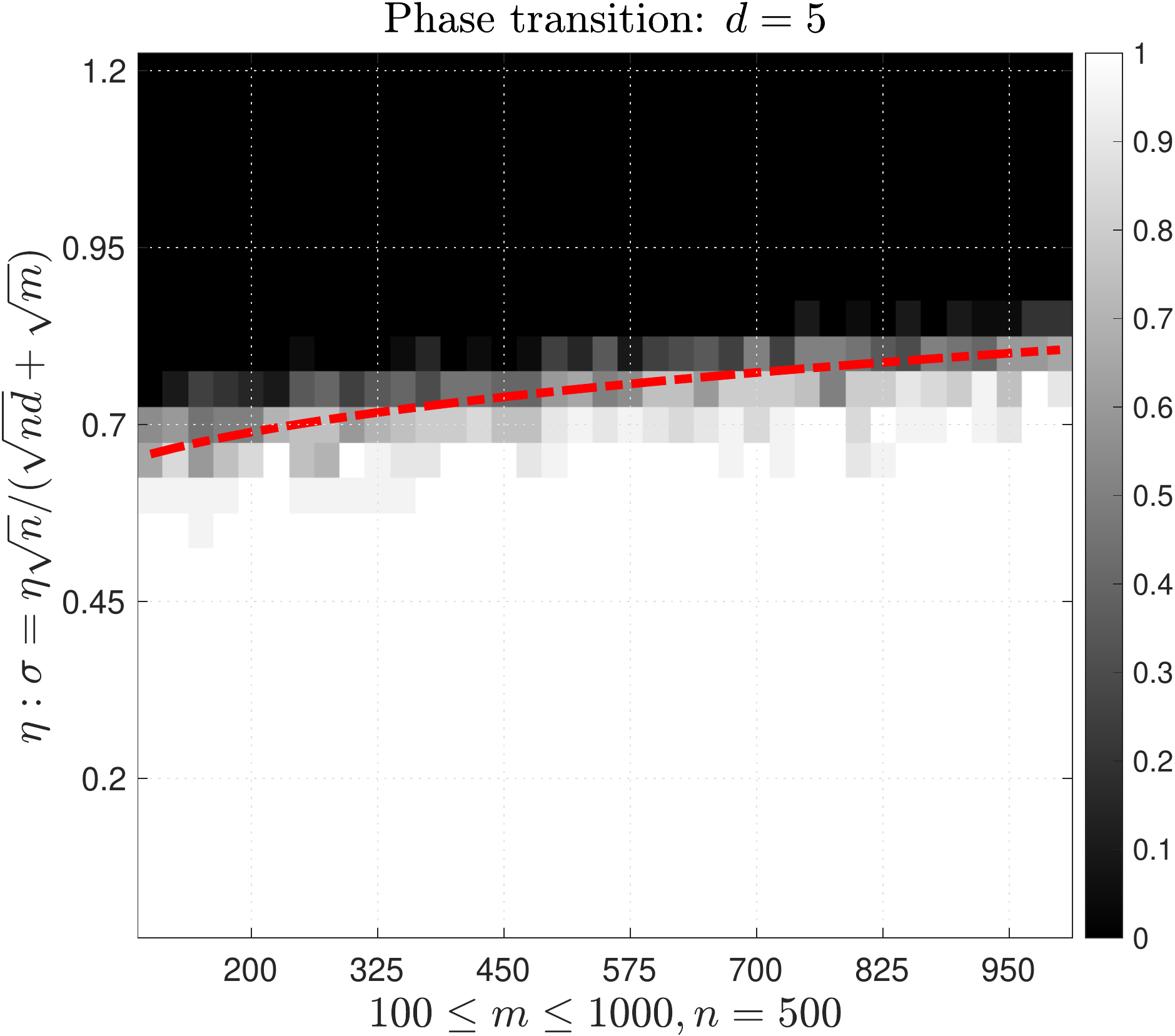}
\end{minipage}
\caption{Phase transition for the tightness/global convergence of the GPM. Black region: the tightness of the~\eqref{def:gpm} fails; white region: the tightness holds. Top: the parameters $(m,d)$ are fixed and $n$ varies; Bottom: the parameters $(n,d)$ are fixed and $m$ varies. The phase transition occurs approximately around the red dashed curve $\sigma = 1.89\sqrt{n}(\sqrt{nd}+\sqrt{m}+2\sqrt{n\log n})$. }
\label{fig:2}
\end{figure}

Figure~\ref{fig:2} presents the simulation experiments with various choices of $(d,n,m)$. 
We can see that the phase transition boundary between the black region (failure) and the white one (success) is quite sharp. In particular, we find the phase transition boundary is approximately 
\[
\sigma \approx \frac{1.89\sqrt{n}}{\sqrt{nd}+\sqrt{m} + 2\sqrt{n\log n}} 
\]
instead of $\sigma\approx \eta^*\sqrt{n}/(\sqrt{nd}+\sqrt{m})$ for some constant $\eta^*$. This means our theoretical bound~\eqref{cond:sigma} differs from the empirical phase transition boundary by a dimension factor $\sqrt{d}$ and a constant factor.

\subsection{Dependence on the condition number}

Note that our bound on $\sigma$ depends on the condition number $\kappa$ of $\BA$ in Theorem~\ref{thm:main}, i.e., for larger $\kappa$, we allow lower noise level for the tightness of~\eqref{def:sdp} and the convergence of~\eqref{def:gpm}. In order to check whether this dependence on $\kappa$ is necessary, we perform the following simulation: we simulate $\BA$ whose condition number ranges $\kappa$ from 1 to 10. In particular, we always let $\sigma_{\min}(\BA)=1$. Then we apply the generalized power method to this synthetic dataset and see if this algorithm  with spectral initialization would have the global convergence. For each pair of $(\kappa, \eta)$ with $1\leq \kappa\leq 10$ and $0\leq \eta\leq 2$, we run 20 experiments and calculate how many instances have the global convergence of the~\eqref{def:gpm}.
\begin{figure}[h!] 
\centering
\includegraphics[width=80mm]{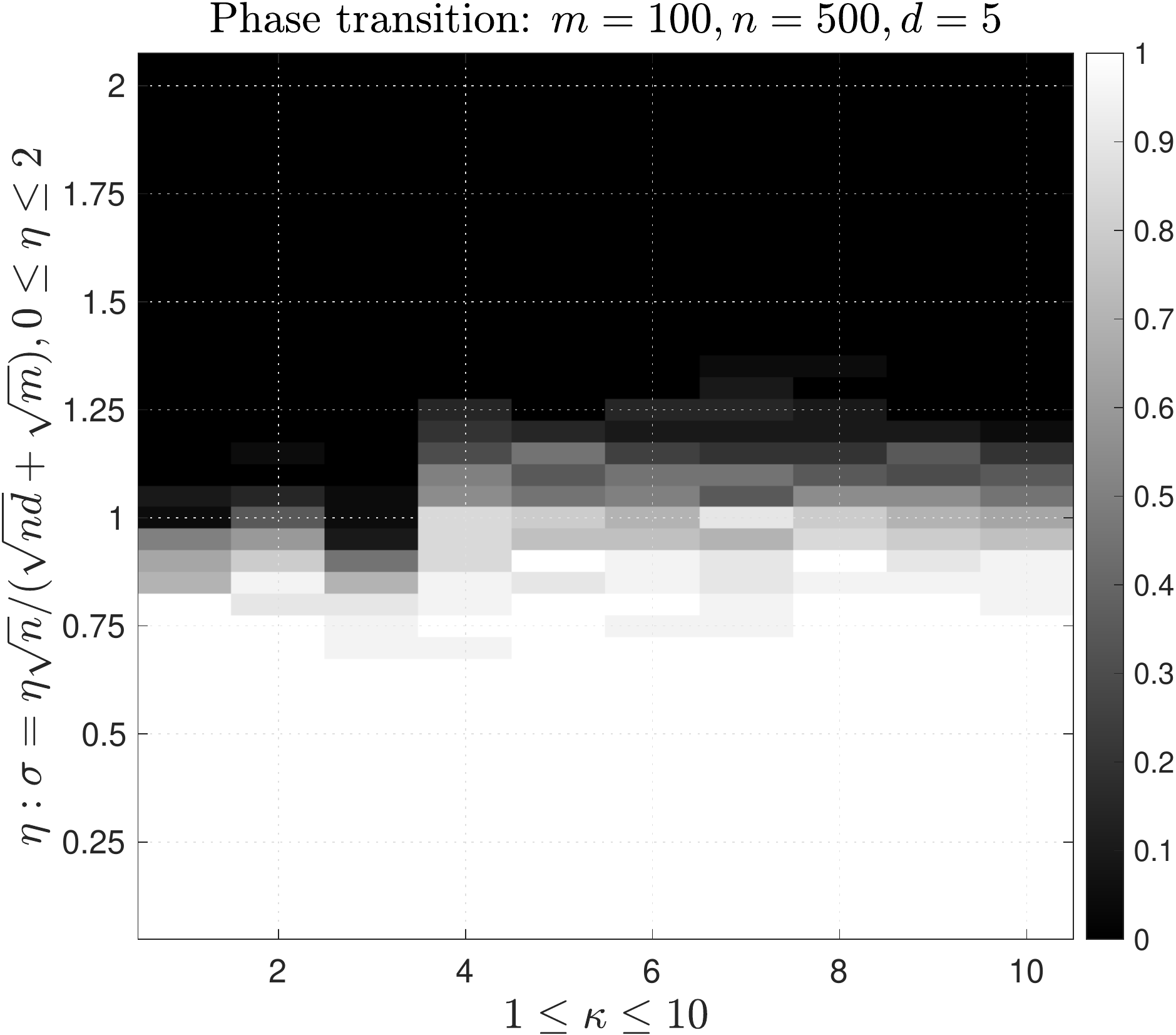}%{pt_d5_fixn500_fixm100_gaussian_kappa_init.eps}
\caption{Dependence of $\sigma$ on the condition number $\kappa$. }
\label{fig:3}
\end{figure}

The result is given in Figure~\ref{fig:3}: the boundary between the success (tightness holds) and failure stays almost constant (i.e., around $\eta\approx 1$) with respect to the condition number $\kappa.$ This indicates almost no dependence of $\sigma$ (i.e., equivalently $\eta$) on $\kappa$, and also that the appearance of $\kappa$ in~\eqref{cond:sigma} is likely an artifact of the proof. We expect that $\sigma$ may only depend on the smallest singular value $\sigma_{\min}(\BA)$ of $\BA$, instead of on $\kappa.$

\section{Proofs}\label{s:proof}

The convergence proof of the GPM for the generalized orthogonal Procrustes problem is inspired by the~\emph{leave-one-out} techniques used in group synchronization problem~\cite{L20c,ZB18}. Due to the inherent setting difference in these two problems, the technical parts differ. The major difference is the construction of the basin of attraction and also different types of additive noise. In our proof, it is more convenient to write the data matrix $\BC = \BD\BD^{\top}$ in~\eqref{def:fo} into the signal-plus-noise form:
\begin{equation}\label{def:SN}
\BC = \underbrace{\BZ\BA\BA^{\top}\BZ^{\top}}_{\text{signal}} + \underbrace{\sigma\BW\BA^{\top}\BZ^{\top} + \sigma \BZ\BA\BW^{\top} + \sigma^2\BW\BW^{\top}}_{\text{noise}}
\end{equation}
where $\BD = \BZ\BA +\sigma\BW\in\RR^{nd\times m}$ is given in~\eqref{def:Z}.

Similar to~\cite{L20c,ZB18}, we decompose the generalized power method into a composition of two operations: 
\begin{equation}\label{def:T}
\BS^{k+1} = {\cal T}(\BS^t)
\end{equation}
where ${\cal T} := \PP_n\circ {\cal L}$ with $\PP_n$ defined in~\eqref{def:Pn} and 
\[
{\cal L}:\RR^{nd\times d}\rightarrow\RR^{nd\times d}, \qquad {\cal L} \BS := \frac{\BC\BS}{n\|\BA\|^2}.
\]

\subsection{An overview of the convergence analysis: the basin of attraction and leave-one-out technique}
To show the global convergence of the GPM, it suffices to understand ${\cal T}$. In particular, it would be ideal to have ${\cal T}$ as a contraction mapping and then linear convergence follows directly. However, this is unclear yet except if the noise is sufficiently small as shown in~\cite{L21a}. Instead, we will show that if $\BS^0$ is initialized in a basin of attraction, i.e., sufficiently close to the ground truth, $\{\BS^t\}_{t\geq 0}$ is a Cauchy sequence in the basin of attraction under the metric $d_F(\cdot,\cdot)$ in~\eqref{def:df}. Then we argue that the limiting point $\BS^{\infty}$ and $\BS^{\infty}(\BS^{\infty})^{\top}$ must be the global maximizer to~\eqref{def:od} and~\eqref{def:sdp} respectively.

\vskip0.25cm 

Throughout our analysis, we will use the noise matrix $\BDelta$: 
\begin{equation}\label{def:delta}
\BDelta := \BW\BA^{\top}\BZ^{\top} +  \BZ\BA\BW^{\top} + \sigma\BW\BW^{\top}\in\RR^{nd\times nd}
\end{equation}
and then $\BC = \BZ\BA\BA^{\top}\BZ^{\top} + \sigma\BDelta.$ In addition, we assume that $\sigma$ satisfies $\sigma\|\BW\| \leq \sqrt{n}\|\BA\|$ which holds automatically under~\eqref{cond:sigma} and then
\begin{equation}\label{eq:normdelta}
\|\BDelta\| \leq 2\sqrt{n}\|\BA\|\|\BW\| + \sigma\|\BW\|^2 \leq 3\sqrt{n}\|\BA\|\|\BW\|.
\end{equation}

The basin of attraction is defined by 
\begin{align*}
\mathcal{N}_{\eps} & : =  \{\BS\in\Od(d)^{\otimes n}: d_F(\BS,\BZ)\leq \eps\sqrt{nd}\}, \\
\mathcal{N}_{\xi,\infty} & = \left\{ \BS\in\Od(d)^{\otimes n}: \max_{1\leq i\leq n}\|\BW_i\BW^{\top}\BS\|_F \leq \xi \sqrt{d}(\sqrt{nd} + \sqrt{m}+ \sqrt{2\gamma n\log n})^2 \right\},
\end{align*}
where we set $\eps,\xi$, and $\gamma$ as 
\begin{equation}\label{eq:par}
\eps=\frac{1}{32\kappa^2\sqrt{d}}, \quad \xi = 6, \quad \gamma\geq 2.%\xi =2( \theta+2), \quad \theta \leq \frac{\sqrt{n}\eps}{2}.
\end{equation}

The first set $\mathcal{N}_{\eps}$ consists of all matrices $\BS$ in $\Od(d)^{\otimes n}$ which are $\eps$-close to the ground truth $\BZ.$ For any $\BS\in\mathcal{N}_{\eps}$, we actually have
\begin{equation}\label{eq:neps_sigma}
d_F^2(\BS,\BZ) = 2nd - 2\|\BZ^{\top}\BS\|_* = 2\sum_{k=1}^d (n - \sigma_{k}(\BZ^{\top}\BS)) \leq \eps^2nd \Longrightarrow \sigma_{\min}(\BZ^{\top}\BS) \geq n(1-\eps^2d/2)
\end{equation}
since $\sigma_k(\BZ^{\top}\BS)\leq n$ for all $1\leq k\leq d.$

The motivation of introducing $\mathcal{N}_\eps$ and $\mathcal{N}_{\xi,\infty}$ is that if $\BS\in\mathcal{N}_{\eps}\cap\mathcal{N}_{\xi,\infty}$, then
\begin{equation}\label{eq:inco}
\max_{1\leq i\leq n}\sigma\|\BDelta_i^{\top}\BS\|_F \leq n
\end{equation}
provided that $\sigma$ satisfies~\eqref{cond:sigma} where $\BDelta_i$ is the $i$th block column of $\BDelta$
\begin{equation}\label{eq:Deltai}
\BDelta_i = \BW\BA^{\top} + \BZ\BA\BW_i^{\top} + \sigma\BW\BW_i^{\top}\in\RR^{nd\times d}.
\end{equation}
The condition~\eqref{eq:inco} is crucial since it ensures for any two elements in $\mathcal{N}_{\eps}\cap \mathcal{N}_{\xi,\infty}$, their distance shrinks after applying ${\cal T}$ to them, as shown in Lemma~\ref{lem:deltaS}.

We will later show that the whole sequence from the GPM will stay in $\mathcal{N}_{\eps}\cap\mathcal{N}_{\xi,\infty}.$ The difficulty comes from keeping the iterates $\BS^t$ in $\mathcal{N}_{\xi,\infty}$, i.e., $\|\BW_i\BW^{\top}\BS^t\|_F \lesssim \sqrt{d}(\sqrt{nd}+\sqrt{m})^2$ modulo the log factor, due to the~\emph{statistical dependence} of $\BW_i$ on $\BW^{\top}\BS^t$. Let's take a quick look at the upper bound of $\|\BW_i\BW^{\top}\BS^t\|_F$, by ignoring the possible logarithmic terms for the time being. Note that for the Gaussian random matrices $\BW_i$ and $\BW\in\RR^{nd\times m}$, we have 
\[
\|\BW_i\|\lesssim \sqrt{m}+\sqrt{d}, \qquad \|\BW\|\lesssim \sqrt{nd}+\sqrt{m}
\]
and
\[
\|\BW^{\top}\BS^t\|\leq \|\BW\|\|\BS^t\| \lesssim \sqrt{n}(\sqrt{nd}+\sqrt{m}), \quad \BS^t\in\Od(d)^{\otimes n}.
\]
Suppose $\BY\in\RR^{m\times d}$ is a matrix independent of $\BW_i$, then we would expect
\[
\|\BW_i\BY\|_F \leq \sqrt{d}\|\BW_i\BY\| \lesssim d\|\BY\| 
\] 
which is implied by classical results on Gaussian random matrix, see Lemma~\ref{lem:gauss} and Corollary~\ref{cor:gauss}. If $\BW^{\top}\BS^t$ was independent of $\BW_i$ (which is of course not), then we would have proven $\|\BW_i\BW^{\top}\BS^t\|_F\lesssim \sqrt{d}(\sqrt{nd}+\sqrt{m})^2.$ However, due to the statistical dependence of $\BW_i$ on $\BW^{\top}\BS^t$, we cannot apply the concentration inequality of Gaussian matrix directly to get a tight bound. On the other hand, using a naive bound on $\|\BW_i\BW^{\top}\BS^t\|_F$ gives
\begin{equation}\label{eq:Nxi}
\|\BW_i\BW^{\top}\BS^t\|_F\leq \|\BW_i\|\|\BW\|\|\BS^t\|_F \lesssim \sqrt{nd}(\sqrt{d} + \sqrt{m}) (\sqrt{nd} + \sqrt{m})
\end{equation}
where $\|\BW_i\|$ and $\|\BW\|$ are given by~\eqref{eq:gauss3a} and~\eqref{eq:gauss3b}.
Comparing this bound with $\mathcal{N}_{\xi,\infty}$, we can see that it differs by a factor of $\sqrt{m}$.

\vskip0.25cm

%The main reason is that $\BW_i$ and $\BW^{\top}\BS$ are statistically dependent. Otherwise, we may invoke Lemma~\ref{lem:gauss} and Corollary~\ref{cor:gauss} to provide a sharp bound on $\|\BW_i\BW^{\top}\BS\|_F$.
Therefore, to control $\|\BW_i\BW^{\top}\BS^t\|_F$ tightly, the main challenge is to show that $\BW_i$ and $\BW^{\top}\BS^t$ are ``nearly" independent. This is made possibly by the~\emph{leave-one-out} technique:
the idea is to replace $\BW^{\top}\BS^t$ with $(\BW^{(i)})^{\top}\BS^{i,t}$ where $\BS^{i,t}$ is the GPM sequence from the data matrix
\begin{equation}\label{def:CDi}
\BC^{(i)} = \BD^{(i)} (\BD^{(i)})^{\top}, \quad \BD^{(i)} = \BZ\BA + \sigma\BW^{(i)}
\end{equation}
where the $j$th block of $\BW^{(i)}$ is
\begin{equation}\label{def:Wi}
\BW^{(i)}_j = 
\begin{cases}
\BW_j, & j\neq i, \\
0, & j =i.
\end{cases}
\end{equation}
The noise  $\BW$ and $\BW^{(i)}$ only differ by its $i$th block $\BW_i$.
The strategy to approximate $\|\BW_i\BW^{\top}\BS^t\|_F$ via decomposing $\BW^{\top}\BS^t$ into
\begin{equation*}
\|\BW_i\BW^{\top}\BS^t\|_F  \leq \|\BW_i(\BW^{(i)})^{\top}\BS^{i,t}\|_F + \|\BW_i\left(\BW^{\top}\BS^t - (\BW^{(i)})^{\top}\BS^{i,t}\right)\|_F.
\end{equation*}
Then since $\BW_i$ and $(\BW^{(i)})^{\top}\BS^{i,t}\in\RR^{m\times d}$ are independent, conditioned on $\|(\BW^{(i)})^{\top}\BS^{i,t}\|\leq \sqrt{n}\|\BW^{(i)}\|$, Corollary~\ref{cor:gauss} implies that
\[
\|\BW_i(\BW^{(i)})^{\top}\BS^{i,t}\|_F\lesssim \sqrt{n}d\|\BW^{(i)}\| \lesssim \sqrt{n}d(\sqrt{nd}+\sqrt{m}).
\]

For the second term in the estimation of $\|\BW_i\BW^{\top}\BS^t\|_F $, we try to bound it by
\begin{align*}
\|\BW_i\left(\BW^{\top}\BS^t - (\BW^{(i)})^{\top}\BS^{i,t}\right)\|_F & \leq \|\BW_i\| \| \BW^{\top}\BS^t - (\BW^{(i)})^{\top}\BS^{i,t} \|_F \\
& \leq \|\BW_i\| \left(  \| \BW^{\top}(\BS^t - \BS^{i,t})) \|_F + \|( \BW-\BW^{(i)})^{\top}\BS^{i,t} \|_F \right).
\end{align*}
We can see that if $\BS^t$ and $\BS^{i,t}$ are close, then the second term is also small.

Our proof will heavily rely on the following useful facts about Gaussian random matrix, the proof of which can be found in~\cite[Theorem 2.26]{W19},~\cite[Theorem 7.3.1]{V18}, and~\cite[Chapter 1]{LT91}.
\begin{lemma}\label{lem:gauss}
For an $n\times m$ $\BX$ with independent $\mathcal{N}(0,1)$ entries, 
then
\begin{equation}
\E \|\BX\|\leq \sqrt{n}+\sqrt{m}
\end{equation}
and
\begin{equation}
\Pr( \left|\|\BX\| - \E\|\BX\| \right| \geq t ) \leq 2\exp(-t^2/2).
\end{equation}
\end{lemma}

A simple and useful corollary of Lemma~\ref{lem:gauss} is written as follows.
\begin{corollary}\label{cor:gauss}
For any matrix $\BY\in\RR^{m\times p}$ which is either independent of $\BX$ or deterministic where $\BX$ is an $n\times m$ Gaussian random, we have
\[
\| \BX\BY \| \leq \|\BY\| (\sqrt{n} + \sqrt{r}+t)
\]
with probability at least $1 - 2\exp(-t^2/2)$ where $r = \rank(\BY)\leq \min\{m,p\}.$
\end{corollary}
The proof follows directly from using the economical SVD of $\BY$ and then applying Lemma~\ref{lem:gauss}.
An immediate application of Lemma~\ref{lem:gauss} and Corollary~\ref{cor:gauss} to the Gaussian noise matrices $\BW_i$ and $\BW$ gives
\begin{align}%\label{eq:gauss3}
\max_{1\leq i\leq n}\|\BW_i\| & \leq \sqrt{d} + \sqrt{m} +\sqrt{2 \gamma\log n}, \label{eq:gauss3a} \\
\|\BW\| & \leq \sqrt{nd} + \sqrt{m} +  \sqrt{2 \gamma\log n},\label{eq:gauss3b} \\
\max_{1\leq i\leq n}\|\BW^{(i)}\| & \leq \sqrt{nd} + \sqrt{m} +  \sqrt{2 \gamma\log n}, \label{eq:gauss3c} \\
\max_{1\leq i\leq n}\|\BW_i\BA^{\top}\| & \leq (2\sqrt{d} +  \sqrt{2 \gamma\log n})\|\BA\|, \label{eq:gauss3d}
\end{align}
with probability at least $1-10n^{-\gamma+1}$ for $\gamma>1$. %For simplicity, we may use $\|\BW\|\leq 3(\sqrt{nd}+\sqrt{m})$ which holds with high probability.

\subsection{Contraction mapping}

This section is devoted to proving ${\cal T} = \PP_n\circ {\cal L}$ in~\eqref{def:T} is a contraction mapping which is similar to~\cite{L20c}. The proof consisting of showing that ${\cal L}$ is a contraction mapping and $\PP_n$ is Lipschitz on $\mathcal{N}_{\eps}\cap\mathcal{N}_{\xi,\infty}$.

\begin{lemma}[${\cal L}$ is a contraction mapping]\label{lem:L}
For $\BX$ and $\BY$ in $\mathcal{N}_{\eps}$, we have
\begin{align*}
d_F({\cal L}\BX, {\cal L}\BY) & \leq \left( 2\eps\sqrt{d} + \frac{3\sigma\|\BW\|}{\sqrt{n}\|\BA\|}  \right)d_F(\BX,\BY).
\end{align*}
Under $\eps= 1/(32\kappa^2\sqrt{d})$,~\eqref{cond:sigma}, and~\eqref{eq:gauss3b},
we have
\begin{equation}\label{eq:L2}
d_F({\cal L}\BX, {\cal L}\BY) \leq \frac{1-\eps^2 d}{2\kappa^2}d_F(\BX,\BY).
\end{equation}
\end{lemma}
\begin{proof}
Let $\BQ = \PP(\BY^{\top}\BX)$ be the orthogonal matrix which satisfies $d_F(\BX,\BY) = \|\BX - \BY\BQ\|_F.$ It holds that
\begin{align*}
d_F({\cal L}\BX, {\cal L}\BY) & \leq \frac{1}{n\|\BA\|^2} \left\| \BZ\BA\BA^{\top}\BZ^{\top}(\BX - \BY\BQ) + \sigma\BDelta(\BX-\BY\BQ)\right\|_F \\
& \leq \frac{1}{\sqrt{n}} \|\BZ^{\top}(\BX- \BY\BQ)\|_F + \frac{\sigma}{n\|\BA\|^2} \|\BDelta(\BX-\BY\BQ)\|_F \\
& = \left( 2\eps\sqrt{d} + \frac{3\sigma\|\BW\|}{\sqrt{n}\|\BA\|}  \right)d_F(\BX,\BY)
\end{align*}
where Lemma 4.5 in~\cite{L20c} gives
\[
\|\BZ^{\top}(\BX- \BY\BQ)\|_F \leq 2\eps\sqrt{nd}\cdot d_F(\BX,\BY)
\]
and~\eqref{eq:normdelta} gives $\|\BDelta\|\leq 3\sqrt{n}\|\BA\|\|\BW\|$. 

Under $\eps= 1/(32\kappa^2\sqrt{d})$,~\eqref{cond:sigma}, and~\eqref{eq:gauss3b},
\[
2\eps\sqrt{d} +  \frac{3\sigma\|\BW\|}{\sqrt{n}\|\BA\|}  \leq \frac{1}{16\kappa^2} + \frac{3\sigma(\sqrt{nd}+\sqrt{m}+\sqrt{2\gamma \log n})}{\sqrt{n}\|\BA\|} \leq \frac{1-\eps^2 d}{2\kappa^2}
\]
where
\[
\sigma\lesssim \frac{\sigma_{\min}(\BA)}{\kappa^2} \cdot \frac{\sqrt{n}}{\sqrt{d}(\sqrt{nd}+\sqrt{m}+\sqrt{2\gamma n\log n})},
\]
which gives~\eqref{eq:L2}.
\end{proof}

\begin{lemma}\label{lem:Pcon}
For any two $d\times d$ matrices $\BX$ and $\BY$, 
\[
\|\PP(\BX) - \PP(\BY)\| \leq \frac{2\|\BX-\BY\|}{\sigma_{\min}(\BX) + \sigma_{\min}(\BY)}
\]
and
\[
\|\PP(\BX) - \PP(\BY)\|_F \leq \frac{2\|\BX-\BY\|_F}{\sigma_{\min}(\BX) + \sigma_{\min}(\BY)}.
\]
\end{lemma}
Lemma~\ref{lem:Pcon} implies that $\PP(\cdot)$ is a Lipschitz continuous function over square matrices whose smallest singular value is bounded away from 0. Its original proof can be found in~\cite{L95} and another proof is in~\cite{L20c}. 

\begin{lemma}\label{lem:Pncon}
For any $\BX$ and $\BY\in\RR^{nd\times d}$, then
\[
d_F(\PP_n(\BX), \PP_n(\BY)) \leq \frac{2d_F(\BX,\BY)}{ \min_{1\leq i\leq n}  \{\sigma_{\min}(\BX_i) + \sigma_{\min}(\BY_i)\}}
\]
where $\BX^{\top} = [\BX_1^{\top},\cdots,\BX_n^{\top}]\in\RR^{d\times nd}$ and $\BY^{\top} = [\BY_1^{\top},\cdots,\BY_n^{\top}]\in\RR^{d\times nd}.$
\end{lemma}
\begin{proof}
The proof directly follows from applying  Lemma~\ref{lem:Pcon} to each $d\times d$ block of $\BX$ and $\BY.$
\end{proof}

In order to show that ${\cal T}$ is a contraction mapping, it suffices to show that each block $[{\cal L}\BX]_i$ with $\BX\in\mathcal{N}_{\eps}\cap\mathcal{N}_{\xi,\infty}$ has smallest singular value bounded away from 0 and then combine Lemma~\ref{lem:L} with Lemma~\ref{lem:Pncon}.

\begin{proposition}\label{prop:con}
For any $\BS\in\mathcal{N}_{\xi,\infty}$, we have
\begin{equation}\label{lem:deltaS}
\|\BDelta_i^{\top}\BS\| \leq 3 \sqrt{nd}(\sqrt{nd} + \sqrt{m} + \sqrt{2\gamma n\log n})\|\BA\|.
\end{equation}
In addition, if $\BS\in\mathcal{N}_{\eps}\cap\mathcal{N}_{\xi,\infty}$, it holds that
\begin{equation}\label{eq:sigmamin}
\sigma_{\min}\left([{\cal L}\BS]_i\right) \geq \frac{1 - \eps^2 d}{\kappa^2}
\end{equation}
provided that 
\[
\sigma\leq \frac{\eps^2 d}{6\kappa^2}\cdot\frac{\sqrt{n}\|\BA\|}{\sqrt{d}(\sqrt{nd} + \sqrt{m}+\sqrt{2\gamma n\log n})}.
\]

As a result, for any $\BX$ and $\BY\in\mathcal{N}_{\eps}\cap\mathcal{N}_{\xi,\infty}$, we have
\begin{equation}\label{eq:contra}
d_F({\cal T}(\BX), {\cal T}(\BY)) \leq \frac{\kappa^2}{1-\eps^2d} \left( 2\eps\sqrt{d} + \frac{3\sigma\|\BW\|}{\sqrt{n}\|\BA\|}\right) d_F(\BX,\BY) \leq \frac{1}{2}d_F(\BX,\BY)
\end{equation}
by combining Lemma~\ref{lem:Pncon} with~\eqref{eq:sigmamin} and~\eqref{eq:L2}.
\end{proposition}

\begin{proof}
For the $i$th block column $\BDelta_i$ in~\eqref{eq:Deltai} of $\BDelta,$ we have
\begin{align*}
\BDelta_i^{\top}\BS & = ( \BW\BA^{\top} + \BZ\BA\BW_i^{\top} + \sigma\BW\BW_i^{\top})^{\top}\BS \\
& = \BA\BW^{\top}\BS + \BW_i\BA^{\top}\BZ^{\top}\BS + \sigma\BW_i\BW^{\top}\BS.
\end{align*}

Then it holds that
\begin{align*}
\|\BDelta_i^{\top}\BS\|_F & \leq \| \BA\BW^{\top}\BS + \BW_i\BA^{\top}\BZ^{\top}\BS + \sigma\BW_i\BW^{\top}\BS \|_F \\
& \leq \sqrt{d}(\|\BA\|\|\BW\|\|\BS\| + \|\BW_i\BA^{\top}\|\|\BZ^{\top}\BS\|) + \sigma\|\BW_i\BW^{\top}\BS\|_F \\
& \leq \sqrt{nd}(\sqrt{nd} + \sqrt{m} + \sqrt{2\gamma\log n})\|\BA\| + n\sqrt{d}(2\sqrt{d}+\sqrt{2\gamma\log n})\|\BA\| \\
& \qquad + \sigma\xi\sqrt{d}(\sqrt{nd} + \sqrt{m} + \sqrt{2\gamma n\log n})^2 \\
& \leq  3 \sqrt{nd}(\sqrt{nd} + \sqrt{m} + \sqrt{2\gamma n\log n})\|\BA\|
\end{align*}
where $\BS\in\mathcal{N}_{\xi,\infty}$, $\xi=6$, $\eps \sqrt{d} \leq 1/(32\kappa^2)$, and
\[
\sigma\leq \frac{\eps^2 d\sqrt{n}\|\BA\|}{6\kappa^2\sqrt{d}(\sqrt{nd} + \sqrt{m} + \sqrt{2\gamma n\log n} )},
\]
and Lemma~\ref{lem:gauss} and Corollary~\ref{cor:gauss} imply
\begin{align*}
& \|\BW\| \leq \sqrt{nd} + \sqrt{m} + \sqrt{2\gamma\log n}, \\
& \max_{1\leq i\leq n}\|\BW_i\BA^{\top}\| \leq (2\sqrt{d}+\sqrt{2\gamma\log n})\|\BA\|
\end{align*}
hold with probability at least $1-O(n^{-\gamma+1})$. 
As a result, we have
\[
\sigma \|\BDelta_i^{\top}\BS\|_F \leq \frac{\eps^2d\cdot \sqrt{n}\|\BA\|}{6\kappa^2\sqrt{d}(\sqrt{nd} + \sqrt{m}+ \sqrt{2\gamma n\log n} )} \cdot  3\sqrt{nd}(\sqrt{nd} + \sqrt{m} + \sqrt{2\gamma n\log n})\|\BA\| \leq \frac{\eps^2 dn\|\BA\|^2}{2\kappa^2}.
\]

For each block of ${\cal L}\BS$ where $\BS\in \mathcal{N}_{\eps}$, we have
\[
[{\cal L}\BS]_i = \frac{1}{n\|\BA\|^2}\left(\BA\BA^{\top}\BZ^{\top}\BS + \sigma\BDelta_i^{\top}\BS\right)
\]
and for $1\leq i\leq n$, it holds that
\[
\sigma_{\min}([{\cal L}\BS]_i) \geq \frac{\kappa^2}{n} \sigma_{\min}(\BZ^{\top}\BS) - \frac{\sigma \|\BDelta_i^{\top}\BS\|}{n\|\BA\|^2}.
\]
Suppose that $\BS\in {\cal N}_{\eps}$, then
\[
\sigma_{\min}(\BZ^{\top}\BS)\geq n\left( 1- \frac{\eps^2 d}{2}\right)
\]
and thus
\[
\sigma_{\min}([{\cal L}\BS]_i) \geq \frac{1}{\kappa^2}\left(1-\frac{\eps^2d}{2}\right) - \frac{\eps^2d}{2\kappa^2} \geq \frac{1-\eps^2 d}{\kappa^2}
\]
which gives~\eqref{eq:sigmamin}.

Now combining~\eqref{eq:sigmamin} with Lemma~\ref{lem:Pncon} leads to
\begin{align*}
d_F({\cal T}(\BX), {\cal T}(\BY)) & \leq \frac{2}{\min_{1\leq i\leq n} \sigma_{\min}([{\cal L}\BX]_i) + \sigma_{\min}([{\cal L}\BY]_i) } d_F({\cal L}\BX, {\cal L}\BY)  \\
& \leq \frac{\kappa^2}{1-\eps^2 d} \cdot d_F({\cal L}\BX, {\cal L}\BY) \\
& \leq \frac{\kappa^2}{1-\eps^2d} \left(2 \eps\sqrt{d} + \frac{3\sigma\|\BW\|}{\sqrt{n}\|\BA\|}\right) d_F(\BX,\BY)
\end{align*}
for any $\BX$ and $\BY$ in $\mathcal{N}_{\eps}\cap\mathcal{N}_{\xi,\infty}.$
\end{proof}

\subsection{Convergence analysis via leave-one-out technique}
Before we analyze the evolution of the GPM, we first show that the spectral estimator $\BS^0$ belongs to $\mathcal{N}_{\eps}\cap\mathcal{N}_{\xi,\infty}$ and is also close to $\BS^{i,0}$, i.e., the spectral estimator associated with the data matrix $\BD^{(i)}$ (or equivalently $\BC^{(i)}$) in~\eqref{def:CDi}.

 %From now on, we set the parameters:
%\[
%\theta \leq \frac{\sqrt{n}\eps}{2}, \quad \xi =2( \theta+2).
%\]
%For simplicity, we can choose $\theta = \eps$ and $\xi = 2(\theta+2) < 6$.

\begin{lemma}[{\bf Initialization}]\label{lem:init}
The spectral initialization~\eqref{def:init} produces $\BS^0$ such that
\[
\min_{\BQ\in\Od(d)}\max_{1\leq i\leq n}\| \BS^0_i - \BQ\|  \leq \frac{120\kappa^2\sigma (\sqrt{nd} + \sqrt{m} + \sqrt{2\gamma n\log n})  }{\sqrt{n}\sigma_{\min}(\BA)}. 
\]
Moreover, we have
\begin{align}
d_F(\BS^0, \BZ) & \leq \frac{\eps\sqrt{nd}}{2}, \label{eq:init1} \\
d_F(\BS^0, \BS^{i,0}) & \leq \eps\sqrt{d}, ~~\forall 1\leq i\leq n, \label{eq:init2}\\
\|\BW_i\BW^{\top}\BS^0\|_F &\leq \frac{\xi}{2} \sqrt{d} (\sqrt{nd} + \sqrt{m} + \sqrt{2\gamma n\log n})^2, \label{eq:init3}
\end{align}
with probability at least $1-O(n^{-\gamma+1})$ provided that
\[
\sigma \lesssim \frac{\sigma_{\min}(\BA)}{\kappa^4}\cdot\frac{\sqrt{n}}{\sqrt{d}(\sqrt{nd} + \sqrt{m} + \sqrt{2\gamma n\log n})}.
\]
\end{lemma}
We leave the proof of the initialization step in Section~\ref{ss:init}.
Next, we will argue that in the first $T$ steps (for simplicity, we set $T= n$), the auxiliary sequence $\BS^{i,t}$ is close to $\BS^t$ for all $1\leq i\leq n$, and moreover $\BS^t$ stays in $\mathcal{N}_{\eps}\cap \mathcal{N}_{\xi,\infty}$. 
We need the following lemma to achieve this goal.
\begin{lemma}
For $0\leq t\leq T$ and $1\leq i\leq n$, we have
\begin{align}
\|\BW_i (\BW^{(i)})^{\top}\BS^{i,t}\| & \leq (2\sqrt{d} + \sqrt{2\gamma\log n}) \| (\BW^{(i)})^{\top}\BS^{i,t} \| \nonumber \\
& \leq \sqrt{n}(2\sqrt{d} + \sqrt{2\gamma\log n})(\sqrt{nd}+\sqrt{m}+\sqrt{2\gamma\log n}), \label{eq:ws1} \\
\|\BW_i (\BW^{(i)})^{\top}\BS^{i,t}\|_F & \leq \sqrt{nd}(2\sqrt{d} + \sqrt{2\gamma\log n})(\sqrt{nd}+\sqrt{m}+\sqrt{2\gamma\log n}), \nonumber \\
& \leq \sqrt{d}\left(\sqrt{nd} + \sqrt{m} + \sqrt{2\gamma n\log n}\right)^2 \label{eq:ws2}
\end{align}
with probability at least $1 - O(T\cdot n^{-\gamma+1}).$
\end{lemma}
The inequalities~\eqref{eq:ws1} and~\eqref{eq:ws2} follow directly from applying Corollary~\ref{cor:gauss} and the independence between $\BW_i$ and $(\BW^{(i)})^{\top}\BS^{i,t}$.

Now we are ready to show that $\BS^t\in\mathcal{N}_{\eps}\cap \mathcal{N}_{\xi,\infty}$ in the first $T$ steps with high probability by using the leave-one-out technique.

\begin{lemma}\label{lem:key1}
Conditioned on~\eqref{eq:ws2}, we have
\begin{align}
d_F(\BS^t, \BS^{i,t}) & \leq \eps\sqrt{d}, ~~\forall 1\leq i\leq n,  \label{eq:basin1} \\
\|\BW_i\BW^{\top}\BS^t\|_F &\leq \frac{\xi}{2} \sqrt{d}(\sqrt{nd} + \sqrt{m}+\sqrt{2\gamma n\log n})^2, \label{eq:basin2} \\
d_F(\BS^t, \BZ) & \leq \frac{\eps\sqrt{nd}}{2}, \label{eq:basin3} 
\end{align}
for $1\leq t\leq T$ with high probability provided that
\begin{equation}\label{cond:sigma_lem}
\sigma \leq \frac{\delta\|\BA\|}{\kappa^2}\cdot \frac{\sqrt{n}}{\sqrt{nd} + \sqrt{m} + \sqrt{2\gamma n\log n}}, \qquad \delta \leq \frac{(1-\eps^2d)\eps}{12}.
\end{equation}
Given $\BS^t\in \mathcal{N}_{\eps}\cap \mathcal{N}_{\xi,\infty}$, Proposition~\ref{prop:con} implies
\[
d_F(\BS^{k+1}, \BS^t)\leq \frac{1}{2}d_F(\BS^{t}, \BS^{t-1})\leq 2^{-t}d_F(\BS^1,\BS^0)
\]
for $1\leq t\leq T$.
\end{lemma}
\begin{remark}
Note that the condition on $\sigma$ in~\eqref{cond:sigma_lem} is equivalent to requiring
\[
\sigma\leq \frac{(1-\eps^2d)\eps\sqrt{d}\|\BA\|}{12\kappa^2} \cdot \frac{\sqrt{n}}{\sqrt{d}(\sqrt{nd} + \sqrt{m} + \sqrt{2\gamma n\log n})}
\]
and therefore~\eqref{cond:sigma} since $\eps^2d \leq 1/(32\kappa^2)$.
\end{remark}

\begin{proof}

We prove~\eqref{eq:basin1}-\eqref{eq:basin3} by induction. Note that the statements hold automatically for $t=0$. Now we assume~\eqref{eq:basin1}-\eqref{eq:basin3} hold for $t\leq k$, we are about to prove the same hold for $t=k+1$.

\vskip0.25cm

\noindent{\bf Step 1: Proof of $\BS^{i,k}\in \mathcal{N}_{\eps}\cap\mathcal{N}_{\xi,\infty}$ for $1\leq i\leq n.$}

Combining~\eqref{eq:basin3} with~\eqref{eq:basin1} immediately implies
\[
d_F(\BS^{i,k},\BZ) \leq  d_F(\BS^k,\BZ) + d_F(\BS^k,\BS^{i,k}) \leq \frac{\eps\sqrt{nd}}{2} + \eps\sqrt{d} \leq \eps\sqrt{nd}.
\]
Let $\BQ = \PP((\BS^{i,k})^{\top}\BS^k)$ and $d_F(\BS^k, \BS^{i,k}) = \|\BS^k - \BS^{i,k}\BQ\|_F.$ Then 
\begin{align*}
\|\BW_i\BW^{\top}\BS^{i,k}\|_F & \leq \|\BW_i\BW^{\top}\BS^k\|_F + \|\BW_i\BW^{\top}(\BS^k - \BS^{i,k}\BQ)\|_F\\
& \leq  \frac{1}{2}\xi \sqrt{d} (\sqrt{nd} + \sqrt{m} + \sqrt{2\gamma n\log n})^2 + \|\BW_i\|\|\BW\| d_F(\BS^k,\BS^{i,k}) \\
& \leq \frac{1}{2} \xi \sqrt{d} (\sqrt{nd} + \sqrt{m} + \sqrt{2\gamma n\log n})^2  \\
& \qquad +(\sqrt{d} + \sqrt{m}+\sqrt{2\gamma \log n}) (\sqrt{nd} + \sqrt{m}+\sqrt{2\gamma \log n}) \eps\sqrt{d} \\
& = \left( \frac{\xi}{2} + \eps\right)\sqrt{d} (\sqrt{nd} + \sqrt{m} + \sqrt{2\gamma n\log n})^2 \\
&  \leq \xi\sqrt{d} (\sqrt{nd} + \sqrt{m} + \sqrt{2\gamma n\log n})^2
\end{align*}
if $\xi \geq 2\eps.$ Therefore, Proposition~\ref{prop:con} implies 
\[
\sigma_{\min}\left( [{\cal L}\BS^{i,k}]_{\ell}\right) \geq \frac{1-\eps^2d}{\kappa^2}, \quad \forall 1\leq \ell\leq n.
\]

\noindent{\bf Step 2: Proof of~\eqref{eq:basin1}.}

Define
\[
{\cal L}^{(i)} = \frac{\BC^{(i)}}{n\|\BA\|^2}
\]
where $\BC^{(i)}$ is defined in~\eqref{def:CDi}.
Then we have $\BS^{i,k+1} = \PP_n({\cal L}^{(i)} \BS^{i,k} )$

Now applying Lemma~\ref{lem:Pncon} gives
\begin{align*}
d_F(\BS^{k+1}, \BS^{i,k+1}) & \leq \frac{\kappa^2}{1-\eps^2d}\cdot d_F({\cal L}\BS^k, {\cal L}^{(i)} \BS^{i,k})
\end{align*}
since both $\BS^k$ and $\BS^{i,k}$ belong to $\mathcal{N}_{\eps}\cap\mathcal{N}_{\xi,\infty}$, and $\sigma_{\min}([{\cal L}\BS^{i,k}]_{\ell})$ and $\sigma_{\min}([{\cal L}^{(i)}\BS^{i,k}]_{\ell})$ are at least $(1-\eps^2d)/\kappa^2$ for $1\leq \ell\leq n.$
To control $d_F(\BS^{k+1},\BS^{i,k+1})$, it suffices to estimate $d_F({\cal L}\BS^k, {\cal L}^{(i)} \BS^{i,k})$:
\begin{align*}
d_F({\cal L}\BS^k, {\cal L}^{(i)} \BS^{i,k}) & \leq d_F({\cal L}\BS^k, {\cal L} \BS^{i,k})  + d_F({\cal L}\BS^{i,k}, {\cal L}^{(i)} \BS^{i,k}) \\
& \leq \frac{1-\eps^2d}{2\kappa^2}\cdot d_F(\BS^k,\BS^{i,k}) + d_F({\cal L}\BS^{i,k}, {\cal L}^{(i)} \BS^{i,k}).
\end{align*}
where the first term is bounded by Lemma~\ref{lem:L}.
For the second term above, we have 
\begin{align*}
d_F({\cal L}\BS^{i,k}, {\cal L}^{(i)} \BS^{i,k}) 
& \leq \| ({\cal L}- {\cal L}^{(i)}) \BS^{i,k}  \|_F \\
& = \frac{\sigma}{n\|\BA\|^2} \left(  \BZ\BA (\BW - \BW^{(i)})^{\top} + (\BW - \BW^{(i)})(\BZ\BA)^{\top}  \right) \BS^{i,k} \\
& \qquad + \frac{\sigma^2}{n\|\BA\|^2} \left(  \BW\BW^{\top} - \BW^{(i)} (\BW^{(i)})^{\top} \right) \BS^{i,k}
\end{align*}
where
\[
 \BW\BW^{\top} - \BW^{(i)} (\BW^{(i)})^{\top} = \BW (\BW - \BW^{(i)})^{\top} + (\BW - \BW^{(i)})(\BW^{(i)})^{\top}.
\]

Now we look at each block of $({\cal L}- {\cal L}^{(i)}) \BS^{i,k}$:
\begin{align*}
& n\|\BA\|^2[({\cal L}- {\cal L}^{(i)}) \BS^{i,k}]_{\ell}  \\
& =
\begin{cases}
\sigma\BA \BW_i^{\top}\BS^{i,k}_i + \sigma^2 \BW_{\ell}\BW^{\top}_i\BS^{i,k}_i, & \ell\neq i, \\
\sigma\BA \BW_i^{\top}\BS^{i,k}_i +  \sigma^2 \BW_i \BW_i^{\top}\BS^{i,k}_i +\sigma\BW_i \BA^{\top} \BZ^{\top}\BS^{i,k} + \sigma^2 \BW_i (\BW^{(i)})^{\top}\BS^{i,k}, & \ell = i.
\end{cases}
\end{align*}

As a result, we have
\begin{align*}
n\|\BA\|^2\|({\cal L}- {\cal L}^{(i)}) \BS^{i,k} \|_F 
& \leq \sigma \|\BZ\BA\BW^{\top}_i \BS^{i,k}_i\|_F +\sigma \|\BW_i \BA^{\top} \BZ^{\top}\BS^{i,k}\|_F  \\
& \quad + \sigma^2 \| \BW\BW_i^{\top}\BS_i^{i,k}\|_F + \sigma^2 \|\BW_i (\BW^{(i)})^{\top}\BS^{i,k}\|_F \\
& \leq 2\sigma n\sqrt{d}(2\sqrt{d}+\sqrt{2\gamma\log n})\|\BA\|  + 2\sigma^2\sqrt{d} \left( \sqrt{nd} + \sqrt{m} + \sqrt{2\gamma n\log n}\right)^2 \\
& \leq 2n\sqrt{d}(2\sqrt{d}+\sqrt{2\gamma\log n})\|\BA\| \cdot \frac{\delta\|\BA\|}{\kappa^2}\cdot \frac{\sqrt{n}}{\sqrt{nd} + \sqrt{m}+ \sqrt{2\gamma n\log n}} \\
& \qquad + \frac{2\delta^2n\sqrt{d}\|\BA\|^2}{\kappa^4} \leq \frac{6\delta n\sqrt{d} \|\BA\|^2}{\kappa^2}
\end{align*}
where $\delta < 1$ and
\[
\sigma \leq \frac{\delta\|\BA\|}{\kappa^2}\cdot \frac{\sqrt{n}}{\sqrt{nd} + \sqrt{m}+ \sqrt{2\gamma n\log n}}
\]
hold under the assumption.

Here Lemma~\ref{lem:gauss} and Corollary~\ref{cor:gauss} indicate that the following estimates hold with probability at least $1-O(T\cdot n^{-\gamma+1})$:
\begin{align*}
\|\BZ\BA\BW^{\top}_i \BS^{i,k}_i\| & \leq \sqrt{n}(2\sqrt{d}+\sqrt{2\gamma\log n})\|\BA\|, \\
\|\BW_i \BA^{\top} \BZ^{\top}\BS^{i,k}\| & \leq n(2\sqrt{d} + \sqrt{2\gamma \log n})\|\BA\|, \\
\| \BW\BW_i^{\top}\BS_i^{i,k}\| & \leq (\sqrt{nd} + \sqrt{m} + \sqrt{2\gamma \log n})(\sqrt{m} + \sqrt{d} + \sqrt{2\gamma \log n}),  \\
\|\BW_i (\BW^{(i)})^{\top}\BS^{i,k}\| & \leq \sqrt{n}(2\sqrt{d}+\sqrt{2\gamma\log n})(\sqrt{m} + \sqrt{d} + \sqrt{2\gamma \log n}),
\end{align*}
where $\BW_i$ and $(\BW^{(i)})^{\top}\BS^{i,k}\in\RR^{m\times d}$ are independent Gaussian matrices, and $\|(\BW^{(i)})^{\top}\BS^{i,k}\| \leq \sqrt{n}(\sqrt{m}+\sqrt{d} + \sqrt{2\gamma\log n})$.

Therefore, we have
\[
\|({\cal L}- {\cal L}^{(i)}) \BS^{i,k} \|_F  \leq \frac{1}{n\|\BA\|^2}\cdot  \frac{6\delta n\sqrt{d} \|\BA\|^2}{\kappa^2} = \frac{6\delta \sqrt{d}}{\kappa^2}.
\]
Under
\[
\delta \leq \frac{(1-\eps^2d)\eps}{12},
\]
we have
\begin{align*}
d_F({\cal L}\BS^k, {\cal L}^{(i)} \BS^{i,k}) & \leq \frac{1-\eps^2d }{2\kappa^2}d_F(\BS^k,\BS^{i,k}) + d_F({\cal L}\BS^{i,k}, {\cal L}^{(i)} \BS^{i,k}) \\
& \leq \frac{(1-\eps^2d)\eps\sqrt{d}}{2\kappa^2} + \frac{6\delta\sqrt{d}}{\kappa^2} \\
&  \leq \frac{1-\eps^2d}{\kappa^2}\cdot \eps\sqrt{d}.
\end{align*}
Then we have
\[
d_F(\BS^{k+1}, \BS^{i,k+1}) \leq \frac{\kappa^2}{1-\eps^2d} d_F({\cal L}\BS^k, {\cal L}^{(i)} \BS^{i,k}) \leq \eps\sqrt{d}.
\]

\vskip0.25cm

\noindent {\bf Step 3: Proof of~\eqref{eq:basin2}, i.e., $\BS^{k+1}\in \mathcal{N}_{\xi,\infty}.$}

Next we proceed to show that $\BS^{k+1}\in \mathcal{N}_{\xi,\infty}$ which follows from~\eqref{eq:basin1} and~\eqref{eq:ws2}. Let $\BQ = \PP((\BS^{i,k+1})^{\top}\BS^{k+1})$ be the orthogonal matrix such that $d_F(\BS^{k+1}, \BS^{i,k+1}) = \| \BS^{k+1} - \BS^{i,k+1}\BQ \|_F.$
 For any $1\leq i\leq n$, we have
\begin{align*}
\|\BW_i \BW^{\top} \BS^{k+1}\|_F & \leq \|\BW_i (\BW^{\top} \BS^{k+1} - (\BW^{(i)})^{\top} \BS^{i,k+1}\BQ)\|_F  + \|\BW_i (\BW^{(i)})^{\top}\BS^{i,k+1}\|_F  
\end{align*}
For the second term above,~\eqref{eq:ws2} gives $ \|\BW_i (\BW^{(i)})^{\top}\BS^{i,k+1}\|_F  \leq \sqrt{d}(\sqrt{nd} + \sqrt{m} + \sqrt{2\gamma n\log n})^2.$ We will focus on the first term which can be decomposed into
\begin{align*}
\|\BW^{\top} \BS^{k+1} - (\BW^{(i)})^{\top} \BS^{i,k+1}\BQ \|_F& \leq\|\BW^{\top} ( \BS^{k+1} -  \BS^{i,k+1}\BQ)\|_F + \|(\BW-\BW^{(i)})^{\top} \BS^{i,k+1}\BQ\|_F \\
& \leq \|\BW\| \cdot d_F(\BS^{k+1},\BS^{i,k+1}) + \|\BW_i\|_F
\end{align*}
Then
\begin{align*}
& \|\BW_i (\BW^{\top} \BS^{k+1} - (\BW^{(i)})^{\top} \BS^{i,k+1}\BQ)\|_F \\
& \quad \leq \|\BW_i\| \left( \|\BW\| d_F(\BS^{k+1},\BS^{i,k+1}) + \|\BW_i\|_F \right) \\
&\quad \leq \sqrt{d} (\sqrt{m} + \sqrt{d} + \sqrt{2\gamma \log n}) \left( (\sqrt{nd} + \sqrt{m}+\sqrt{2\gamma \log n}) \eps + (\sqrt{m}+\sqrt{d}+\sqrt{2\gamma \log n}) \right) \\
&\quad \leq (\eps+1)\sqrt{d} (\sqrt{nd} + \sqrt{m}+\sqrt{2\gamma \log n})^2
\end{align*}
where $d_F(\BS^{k+1},\BS^{i,k+1})\leq \eps\sqrt{d}$ follows from~\eqref{eq:basin1}.

Now combining these estimates above leads to
\begin{align*}
\|\BW_i \BW^{\top} \BS^{k+1}\|_F & \leq (\eps+1)\sqrt{d} (\sqrt{nd} + \sqrt{m}+\sqrt{2\gamma\log n})^2 + \sqrt{d}(\sqrt{nd} + \sqrt{m}+\sqrt{2\gamma n\log n})^2 \\
& \leq  (\eps+2) \sqrt{d} (\sqrt{nd} + \sqrt{m}+\sqrt{2\gamma n\log n})^2 \\
& \leq \frac{\xi}{2} \sqrt{d} (\sqrt{nd} + \sqrt{m}+\sqrt{2\gamma n\log n})^2
\end{align*}
where $\xi =6 > 2(\eps+2)$.

\vskip0.25cm

\noindent{\bf Step 4: Proof of~\eqref{eq:basin3}, i.e., $\BS^{k+1}\in\mathcal{N}_{\eps}.$}

Finally, we show that $d_F(\BS^{k+1},\BZ)\leq \eps\sqrt{nd}.$ 
Now we invoke Lemma~\ref{lem:Pcon} to estimate $d_F(\BS^{k+1},\BZ)$.
Note that  for $\BA\BA^{\top}\succ 0$, we have
\[
\PP_n\left(\frac{ \BZ\BA\BA^{\top}}{\|\BA\|^2}\right)  = \BZ, \qquad
\sigma_{\min}\left(\frac{[\BZ\BA\BA^{\top}]_i}{\|\BA\|^2} \right) = \sigma_{\min}\left(\frac{\BA\BA^{\top}}{\|\BA\|^2} \right) \geq \frac{1}{\kappa^2}.
\]
Also we have $\sigma_{\min}([{\cal L}\BS^k]_i) \geq (1-\eps^2d)/\kappa^2$, which follows from $\BS^k\in\mathcal{N}_{\eps}\cap\mathcal{N}_{\xi,\infty}$ and~\eqref{eq:sigmamin}. Then
\begin{align*}
d_F(\BS^{k+1},\BZ) 
& = d_F\left(\PP_n({\cal L}\BS^k), \PP_n\left(\frac{\BZ\BA\BA^{\top}}{\|\BA\|^2}\right)\right)  \leq \frac{\kappa^2}{1-\eps^2d/2} \cdot d_F\left({\cal L}\BS^k,\frac{\BZ\BA\BA^{\top}}{\|\BA\|^2}\right).
\end{align*}
Let $\BQ = \PP(\BZ^{\top}\BS^k)$ and we have
\begin{align*}
d_F\left({\cal L}\BS^k,\frac{\BZ\BA\BA^{\top}}{\|\BA\|^2}\right) & \leq \left\| {\cal L}\BS^k - \frac{\BZ\BA\BA^{\top} \BQ}{\|\BA\|^2} \right\|_F \\
& = \frac{1}{n\|\BA\|^2}\left\| \BZ\BA\BA^{\top}\left( \BZ^{\top}\BS^k - n\BQ\right)+ \sigma\BDelta \BS^k \right\|_F \\
(\|\BS^k\|_F=\sqrt{nd})\qquad & \leq \frac{1}{\sqrt{n}}\|\BZ^{\top}\BS^k - n\BQ\|_F + \sigma\sqrt{\frac{d}{n}} \frac{\|\BDelta\|}{\|\BA\|^2}.
\end{align*}

Note that
\[
\|\BZ^{\top}\BS^k-n\BQ\|_F \leq d_F^2(\BS^t,\BZ)\leq \eps^2nd, \quad \|\BDelta\| \leq 3\sqrt{n}\|\BA\|\|\BW\|
\]
which follow from~\eqref{eq:basin3} and~\eqref{eq:normdelta}.
We have
\begin{align*}
d_F(\BS^{k+1},\BZ) & \leq \frac{\kappa^2}{1-\eps^2d/2}\cdot d_F\left({\cal L}\BS^k,\frac{\BZ\BA\BA^{\top}}{\|\BA\|^2}\right) \\
& \leq \frac{\kappa^2}{1-\eps^2d/2}\cdot \left( \eps^2\sqrt{n}d + \sigma \sqrt{\frac{d}{n}}\cdot\frac{3\sqrt{n}\|\BA\|\|\BW\|}{\|\BA\|^2} \right) \\
& \leq \frac{\kappa^2}{1-\eps^2d/2}\left( \eps\sqrt{d} +   \frac{3\sigma\|\BW\|}{\eps\sqrt{n}\|\BA\|}\right)   \eps\sqrt{nd}.
\end{align*} 
Note that $\|\BW\|\leq \sqrt{nd} + \sqrt{m} + \sqrt{2\gamma \log n}$. 
Then it holds
\[
\eps\sqrt{d} +  \frac{3\sigma\|\BW\|}{\eps\sqrt{n}\|\BA\|} \leq \frac{1-\eps^2d/2}{\kappa^2} \Longrightarrow d_F(\BS^{k+1},\BZ) \leq \eps\sqrt{nd}.
\]
provided that $\eps= 1/(32\kappa^2\sqrt{d})$ and 
\[
\sigma \leq \frac{\eps\sqrt{d}}{6\kappa^2}\cdot \frac{ \sqrt{n}\|\BA\|}{\sqrt{d}(\sqrt{nd} + \sqrt{m}+\sqrt{2\gamma n\log n})}\Longrightarrow  \frac{3\sigma\|\BW\|}{\eps\sqrt{n}\|\BA\|} \leq \frac{1}{2\kappa^2}.
\]
Then
\[
d_F(\BS^{k+1},\BZ) \leq \frac{\kappa^2}{1-\eps^2d/2}\cdot \frac{1-\eps^2d/2}{\kappa^2}\eps\sqrt{nd} = \eps\sqrt{nd}.
\]
\end{proof}

\begin{lemma}\label{lem:convt}
Conditioned on the fact that for any $0\leq t\leq n$, 
\begin{align*}
\|\BW_i\BW^{\top}\BS^t\|_F &\leq \frac{\xi}{2} \sqrt{d}(\sqrt{nd} + \sqrt{m}+\sqrt{2\gamma n\log n})^2, \\
d_F(\BS^t, \BS^{i,t}) & \leq \eps\sqrt{d}, ~~\forall 1\leq i\leq n,\\
d_F(\BS^t, \BZ) & \leq \frac{\eps\sqrt{nd}}{2},
\end{align*}
we have $\BS^t\in\mathcal{N}_{\eps}\cap\mathcal{N}_{\xi,\infty}$ and
\[
d_F(\BS^t, \BS^{t-1})\leq \frac{1}{2}d_F(\BS^{t-1}, \BS^{t-2})\leq 2^{-t+1}d_F(\BS^1,\BS^0)
\]
for all $t\geq 0$.
\end{lemma}

\begin{proof}
We will prove the statement by induction. Suppose for $t \leq n + k$ is true, we will prove it is also true for $t= n+k+1.$ Note that this is true for $k=0$: Lemma~\ref{lem:key1} implies $d_F(\BS^{t+1},\BS^t)\leq 2^{-1}d_F(\BS^t,\BS^{t-1})$ for $0\leq t\leq n$ since $\{\BS^{t}\}_{t=0}^n\in\mathcal{N}_{\eps}\cap\mathcal{N}_{\xi,\infty}$.  

Then
\[
d_F(\BS^{n+k+1},\BS^{n+k}) \leq \frac{1}{2}d_F(\BS^{n+k}, \BS^{n+k-1})\leq 2^{-k}d_F(\BS^n, \BS^{n-1})
\]
since $\BS^t\in \mathcal{N}_{\eps}\cap\mathcal{N}_{\xi,\infty}$ for $t\leq n+k.$ 

We start with showing that $d_F(\BS^{n+k+1},\BS^n)\leq d_F(\BS^n,\BS^{n-1})$:
\begin{align*}
d_F(\BS^{n+k+1},\BS^n) & \leq \sum_{\ell=0}^k d_F(\BS^{n+\ell+1},\BS^{n+\ell})   \leq \sum_{\ell=0}^k 2^{-\ell-1} d_F(\BS^{n},\BS^{n-1}) \leq d_F(\BS^n,\BS^{n-1}).
\end{align*}
Also Lemma~\ref{lem:key1} implies that
\[
d_F(\BS^{n+k+1},\BS^n)  \leq d_F(\BS^n,\BS^{n-1}) \leq 2^{-n+1}d_F(\BS^1,\BS^0) \leq 2^{-n+1}\eps\sqrt{nd}
\]
where $d_F(\BS^1,\BS^0) \leq d_F(\BS^1,\BZ)  + d_F(\BS^0,\BZ)\leq \eps\sqrt{nd}.$

Note that
\begin{align*}
d_F(\BS^{n+k+1},\BZ) & \leq d_F(\BS^{n+k+1},\BS^{n}) +  d_F(\BS^{n},\BZ)  \\
& \leq d_F(\BS^{n},\BS^{n-1})  + \frac{\eps\sqrt{nd}}{2}\\
& \leq 2^{-n+1}\cdot \eps\sqrt{nd} +\frac{\eps\sqrt{nd}}{2} \\
& \leq \eps\sqrt{nd},
\end{align*}
which gives $\BS^{n+k+1}\in\mathcal{N}_{\eps}.$

We proceed to show $\BS^{n+k+1}\in\mathcal{N}_{\xi,\infty}$: by letting $\BQ=\PP( (\BS^n)^{\top}\BS^{n+k+1} )$, for any $1\leq i\leq n$, it holds that
\begin{align*}
\|\BW_i\BW^{\top}\BS^{n+k+1}\|_F & \leq \|\BW_i\BW^{\top}(\BS^{n+k+1} - \BS^n\BQ)\|_F +  \|\BW_i\BW^{\top} \BS^n\|_F \\
& \leq \|\BW_i\|\|\BW\|d_F(\BS^{n+k+1},\BS^n) +\|\BW_i\BW^{\top} \BS^n\|_F  \\
~\eqref{eq:gauss3a}\text{ and}~\eqref{eq:gauss3b}\qquad & \leq (\sqrt{m}+\sqrt{d} + \sqrt{2\gamma \log n})(\sqrt{nd} + \sqrt{m} + \sqrt{2\gamma \log n})2^{-n+1}\eps\sqrt{nd} \\
(\BS^n\in\mathcal{N}_{\xi,\infty}) \qquad & \qquad + \frac{1}{2}\xi\sqrt{d}(\sqrt{nd} + \sqrt{m}+\sqrt{2\gamma n\log n})^2 \\
& \leq \xi \sqrt{d} (\sqrt{nd} + \sqrt{m}+\sqrt{2\gamma n\log n})^2
\end{align*}
where $2^{-n+1}\eps\sqrt{n} \leq \eps< \xi$ which gives $\BS^{n+k+1}\in\mathcal{N}_{\xi,\infty}.$

Then using Proposition~\ref{prop:con} leads to
\[
d_F(\BS^{t+1},\BS^{t})\leq \frac{1}{2} d_F(\BS^{t},\BS^{t-1}) \leq 2^{-t+1}d_F(\BS^1,\BS^0)
\]
for $t\leq n+k+1$ and thus for all $t$ by induction.
\end{proof}

The global optimality of the SDP relaxation is characterized by the following theorem. 
\begin{theorem}\label{thm:cvx}
The matrix $\BX = \BS\BS^{\top}$ with $\BS\in \Od(d)^{\otimes n}$ is a global optimal solution to~\eqref{def:sdp} if there exists a block-diagonal matrix $\BLambda\in\RR^{nd\times nd}$ such that
\begin{equation}\label{cond:opt}
\BC\BS = \BLambda \BS, \quad \BLambda - \BC\succeq 0.
\end{equation}
In addition, if $\BLambda - \BC$ is of rank $(n-1)d$, then $\BX$ is the unique global maximizer. 
\end{theorem}
This optimality condition can be found in several places including~\cite[Proposition 5.1]{L20a} and~\cite[Theorem 7]{RCBL19}.
The derivation of Theorem~\ref{thm:cvx} follows from the standard routine of duality theory in convex optimization. 

%\begin{lemma}\label{lem:globalopt}
%The sequence $\{\BS^t\}_{t\geq 0}$ converges to a unique limiting point $\BS^{\infty}$. Moreover, $\BS^{\infty}$ is the unique global optimal solution to~\eqref{def:od} and $\BS^{\infty}(\BS^{\infty})^{\top}$ is the unique optimal solution to the SDP relaxation~\eqref{def:sdp}. 
%\end{lemma} Proof of Theorem~\ref{thm:error}
\begin{proof}[\bf Proof of Theorem~\ref{thm:main} and~\ref{thm:error}]
First, following from Lemma~\ref{lem:convt}, we have $d_F(\BS^t,\BS^{t-1})\leq 2^{-t+1}d_F(\BS^1,\BS^0).$ 
It is easy to see that $\{\BS^t\}_{t\geq 0}$ is a Cauchy sequence under $d_F(\cdot,\cdot)$ on the closed and compact set $\Od(d)^{\otimes n}.$ More precisely, for any $t'> t$, we have
\[
d_F(\BS^t,\BS^{t'}) \leq \sum_{\ell=0}^{t'-t-1} d_F(\BS^{t+\ell},\BS^{t+\ell+1}) \leq \sum_{\ell=0}^{t'-t-1} 2^{-\ell} d_F(\BS^{t+1},\BS^t) \leq 2d_F(\BS^{t+1},\BS^t) 
\]
which converges to 0 as $t\rightarrow\infty$. Therefore, there exists a unique limiting point $\BS^{\infty}$. Note that $\{\BS^t\}$ belong to the compact set $\mathcal{N}_{\eps}\cap\mathcal{N}_{\xi,\infty}$, and so is $\BS^{\infty}$. Now we will show that $\BS^{\infty}$ is also a fixed point of the nonlinear mapping ${\cal T}$.
\begin{align*}
d_F({\cal T}(\BS^{\infty}), \BS^{\infty}) & \leq d_F({\cal T} (\BS^{\infty}), {\cal T}(\BS^t) ) + d_F(\BS^{t+1}, \BS^{\infty}) \\
&\leq \frac{3}{2}d_F(\BS^{t+1},\BS^{\infty}) \longrightarrow 0
\end{align*}
as $t\rightarrow\infty$ since $\{\BS^t\}$ is a Cauchy sequence. As a result, we have $d_F({\cal T}(\BS^{\infty}), \BS^{\infty}) = 0$ holds and it implies ${\cal T}(\BS^{\infty}) = \BS^{\infty}\BQ$ for some orthogonal matrix $\BQ.$ 

Now we will show $\BQ=\I_d$.
Note that we have $\BS^{\infty}\in\mathcal{N}_{\eps}\cap\mathcal{N}_{\xi,\infty}$, and thus $\sigma_{\min}([{\cal L}\BS^{\infty}]_i) > 0$ for all $1\leq i\leq n$, following from~\eqref{eq:sigmamin}. It immediately holds
\[
\PP([{\cal L}\BS^{\infty}]_i)=([{\cal L}\BS^{\infty}]_i [{\cal L}\BS^{\infty}]_i^{\top} )^{-1/2} [{\cal L}\BS^{\infty}]_i  = \BS_i^{\infty}\BQ 
\]
which equals
\[
{\cal L}\BS^{\infty} = \BLambda\BS^{\infty}\BQ
\]
for some block diagonal matrices $\BLambda$ whose $i$th diagonal block $\BLambda_{ii}$ equals $ ([{\cal L}\BS^{\infty}]_i [{\cal L}\BS^{\infty}]_i^{\top})^{1/2}$ and $\BLambda_{ii}\succ 0.$
 Then
\[
(\BS_i^{\infty})^{\top} [{\cal L}\BS^{\infty}]_i = \BLambda_{ii} \BS^{\infty}_i\BQ \Longrightarrow (\BS^{\infty})^{\top}{\cal L}\BS^{\infty} = \sum_{i=1}^n ( \BS^{\infty}_i)^{\top} \BLambda_{ii} \BS^{\infty}_i\BQ
\]
with $( \BS^{\infty}_i)^{\top} \BLambda_{ii} \BS^{\infty}_i\succ 0.$

Direct computation gives that
\begin{align*}
(\BS^{\infty})^{\top}{\cal L}\BS^{\infty}= \frac{1}{n\|\BA\|^2} (\BS^{\infty})^{\top} \left(  \BZ\BA\BA^{\top}\BZ^{\top} + \sigma\BDelta \right) \BS^{\infty}
\end{align*}
and then
\begin{align*}
\lambda_{\min}((\BS^{\infty})^{\top}{\cal L}\BS^{\infty}) & \geq \frac{\lambda_{\min}(\BA\BA^{\top})}{n\|\BA\|^2} \sigma_{\min}^2(\BZ^{\top}\BS^{\infty}) - \frac{\sigma\|\BDelta\| }{\|\BA\|^2}  \\
& \geq \frac{n}{\kappa^2}\left( 1- \frac{\eps^2d}{2}\right)^2 - \frac{3\sigma \sqrt{n}\|\BW\|}{\|\BA\|} \\
& \geq \frac{n}{\kappa^2}\left( 1- \frac{\eps^2d}{2}\right)^2 - \frac{3\sigma \sqrt{n}(\sqrt{nd}+\sqrt{m} + \sqrt{2\gamma \log n})}{\|\BA\|}  > 0
\end{align*}
under~\eqref{eq:gauss3b} and
\[
\sigma \leq \frac{(1-\eps^2d/2)^2\|\BA\|}{3\kappa^2}\cdot\frac{\sqrt{n}}{\sqrt{nd} + \sqrt{m} + \sqrt{2\gamma\log n}}.
\]
The only orthogonal matrix which makes 
\[
(\BS^{\infty})^{\top}{\cal L}\BS^{\infty} = \sum_{i=1}^n ( \BS^{\infty}_i)^{\top} \BLambda_{ii} \BS^{\infty}_i\BQ
\]
hold for two strictly positive semidefinite matrices $(\BS^{\infty})^{\top}{\cal L}\BS^{\infty} $ and $\sum_{i=1}^n( \BS^{\infty}_i)^{\top} \BLambda_{ii} \BS^{\infty}_i$ must be $\I_d.$

Now we can conclude that the limiting point satisfies
\[
{\cal T}(\BS^{\infty}) =\BS^{\infty}
\]
which also implies $(\BLambda - {\cal L})\BS^{\infty} = 0.$ The linear convergence of $\BS^t$ to $\BS^{\infty}$ follows directly from~\eqref{eq:contra}:
\[
d_F(\BS^{\infty},\BS^t) = d_F({\cal T}(\BS^{\infty}), {\cal T}(\BS^{t-1})) \leq \frac{1}{2}d_F(\BS^{\infty},\BS^{t-1}) \leq \cdots \leq 2^{-t}d_F(\BS^{\infty},\BS^0)
\] 
since all $\{\BS^t\}_{t\geq 0}$ are in $\mathcal{N}_{\eps}\cap\mathcal{N}_{\xi,\infty}.$

Now to argue $\BS^{\infty}(\BS^{\infty})^{\top}$ is the unique global optimal solution to the SDP, it suffices to show that $\BLambda - {\cal L}\succeq 0$ with ${\cal L} = \BC/(n\|\BA\|^2)$ and the $(d+1)$-th smallest eigenvalue is strictly positive, according to Theorem~\ref{thm:cvx}. Remember that $\lambda_{\min}\left( \BLambda_{ii}\right) \geq (1-\eps^2d)/\kappa^2$ since $\BLambda_{ii} = [{\cal L}\BS^{\infty}]_i [{\cal L}\BS^{\infty}]_i^{\top}$ holds with $\sigma_{\min}([{\cal L}\BS^{\infty}]_i)\geq (1-\eps^2d)/\kappa^2$. It suffices to show that for any unit vector $\bu$ which is perpendicular with all $d$ columns of $\BS^{\infty}$, we have $\bu^{\top}(\BLambda-{\cal L})\bu>0$. First, 
\begin{align*}
\bu^{\top}\BC\bu & \leq \bu^{\top}(\BZ\BA\BA^{\top}\BZ^{\top}+ \sigma\BDelta)\bu \\
& \leq  \|\BA\|^2\|\BZ^{\top}\bu\|^2 + \sigma \|\BDelta\| \\
& \leq \|\BA\|^2\|(\BZ\BQ - \BS^{\infty})^{\top}\bu\|^2 + \sigma \|\BDelta\| \\
& \leq \|\BA\|^2 \eps^2 nd + 3\sqrt{n}\sigma\|\BA\|\|\BW\|
\end{align*}
where $\BQ = \PP(\BZ^{\top}\BS^{\infty}).$
Thus
\[
\bu^{\top}{\cal L}\bu = \frac{1}{n\|\BA\|^2}\bu^{\top}{\cal L}\bu  \leq  \eps^2 d + \frac{3\sigma\|\BW\|}{\sqrt{n}\|\BA\|}
\]
for any unit vector $\bu$ with $\bu^{\top}\BS^{\infty} = 0.$
In conclusion, we have
\begin{align*}
\bu^{\top}(\BLambda-{\cal L})\bu & \geq \lambda_{\min}(\BLambda) - \max_{\|\bu\|=1,\bu^{\top}\BS^{\infty} =0} \bu^{\top}{\cal L}\bu \\
&  \geq \frac{1-\eps^2d}{\kappa^2} - \left(\eps^2 d + \frac{3\sigma\|\BW\|}{\sqrt{n}\|\BA\|}\right) \\
& \geq \frac{0.99}{\kappa^2} - \frac{3\sigma\|\BW\|}{\sqrt{n}\|\BA\|}
\end{align*}
where $\eps\sqrt{d}=1/(32\kappa^2)$.
We need to have
\[
\frac{3\sigma\|\BW\|}{\sqrt{n}\|\BA\|} \leq \frac{0.9}{\kappa^2} \Longleftrightarrow \sigma\leq\frac{0.3\|\BA\|}{\kappa^2}\cdot \frac{\sqrt{n}}{\sqrt{nd}+ \sqrt{m}+\sqrt{2\gamma \log n}}
\]
to ensure the $(d+1)$-th smallest eigenvalue of $\BLambda-{\cal L}$ is strictly positive. As a result, $\BS^{\infty}(\BS^{\infty})^{\top}$ is the unique optimal solution to the~\eqref{def:sdp} relaxation.

\vskip0.25cm

For the proof of Theorem~\ref{thm:error}, i.e., the blockwise error bound of $\BS^{\infty}$, it suffices to consider
\begin{align*}
\left\| [{\cal L}\BS^{\infty}]_i - \frac{\BA\BA^{\top}}{n\|\BA\|^2}\BZ^{\top}\BS^{\infty}   \right\|_F & = \left\| \frac{1}{n\|\BA\|^2}( \BA\BA^{\top}\BZ^{\top}\BS^{\infty} + \sigma\BDelta_i^{\top}\BS^{\infty} ) - \frac{\BA\BA^{\top}}{n\|\BA\|^2}\BZ^{\top}\BS^{\infty} \right\|_F \\
& \leq \frac{\sigma}{n\|\BA\|^2}\|\BDelta_i^{\top}\BS^{\infty}\|_F \\
~\eqref{lem:deltaS}\qquad & \leq \frac{6\sigma \sqrt{d}(\sqrt{nd} + \sqrt{m} + \sqrt{2\gamma n\log n})}{\sqrt{n}\|\BA\|}.
\end{align*}

Note that
\begin{align*}
\sigma_{\min}\left( \frac{\BA\BA^{\top}}{n\|\BA\|^2}\BZ^{\top}\BS^{\infty}  \right) & \geq \frac{\sigma_{\min}(\BZ^{\top}\BS^{\infty})}{n\kappa^2} \geq \frac{1-\eps^2d}{\kappa^2}
\end{align*}
which follows from $\BS^{\infty}\in\mathcal{N}_{\eps}\cap\mathcal{N}_{\xi,\infty}$,~\eqref{eq:neps_sigma} and~\eqref{eq:sigmamin}. Now we let 
\[
\BQ:= \PP_n\left(\frac{\BA\BA^{\top}}{n\|\BA\|^2}\BZ^{\top}\BS^{\infty}  \right).
\]
Then applying Lemma~\ref{lem:Pcon} gives
\begin{align*}
\|\BS_i^{\infty}-\BQ\|_F& =  \| \PP( [{\cal L}\BS^{\infty}]_i) - \BQ \|_F \\
& \leq \frac{\kappa^2}{1-\eps^2d} \left\| [{\cal L}\BS^{\infty}]_i - \frac{\BA\BA^{\top}}{n\|\BA\|^2}\BZ^{\top}\BS^{\infty} \right\|_F \\
& \leq \frac{6\kappa^2}{(1-\eps^2d)\|\BA\|}  \cdot \frac{\sigma \sqrt{d}(\sqrt{nd} + \sqrt{m} + \sqrt{2\gamma n\log n})}{\sqrt{n}}
\end{align*}
where $\sigma_{\min}([{\cal L}\BS^{\infty}]_i)\geq (1-\eps^2d )/\kappa^2.$

Finally, we will estimate $\min_{\BR\in\Od(d)}\|\widehat{\BA} - \BR\BA\|_F$:  it holds
\begin{align*}
\min_{\BR\in\Od(d)}\left\|\widehat{\BA} - \BR\BA\right\|_F 
& \leq \left\| \frac{1}{n} (\BS^{\infty})^{\top}\BD - \BQ^{\top}\BA\right\|_F \\
& \leq \left\| \left(\frac{1}{n} \BZ^{\top} \BS^{\infty}- \BQ\right)^{\top}\BA\right\|_F + \frac{\sigma}{n}\left\| \BW^{\top}\BS^{\infty}\right\|_F
\end{align*}
where $\BD = \BZ\BA + \sigma\BW$ and
\[
\widehat{\BA}  = \frac{1}{n}\sum_{i=1}^n (\BS_i^{\infty})^{\top}\BA_i = \frac{1}{n}(\BS^{\infty})^{\top}\BD.
\]
Note that 
\begin{align*}
\| \BZ^{\top}\BS^{\infty}-n\BQ \|_F & = \|\BZ^{\top}(\BS^{\infty} - \BZ\BQ)\|_F \leq \sqrt{n}\| \BS^{\infty} - \BZ\BQ \|_F \\
& \leq n\max_{1\leq i\leq n}\|\BS_i^{\infty} - \BQ\|_F \\
& \leq \frac{6n\kappa^2}{(1-\eps^2d)\|\BA\|}  \cdot \frac{\sigma \sqrt{d}(\sqrt{nd} + \sqrt{m} + \sqrt{2\gamma n\log n})}{\sqrt{n}}.
\end{align*}
Thus
\begin{align*}
\min_{\BR\in\Od(d)}\left\|\widehat{\BA} - \BR\BA\right\|_F 
& \leq\left\| \left(\frac{1}{n} \BZ^{\top} \BS^{\infty}- \BQ\right)^{\top}\BA\right\|_F + \frac{\sigma}{n}\left\| \BW^{\top}\BS^{\infty}\right\|_F \\
& \leq \frac{\|\BA\|}{n}\cdot \frac{6n\kappa^2}{(1-\eps^2d)\|\BA\|}  \cdot \frac{\sigma \sqrt{d}(\sqrt{nd} + \sqrt{m} + \sqrt{2\gamma n\log n})}{\sqrt{n}}  \\
& \qquad + \frac{\sigma}{n}\cdot \sqrt{nd}(\sqrt{nd}+\sqrt{m}+\sqrt{2\gamma\log n}) \\
& \leq   \frac{8\kappa^2\sigma \sqrt{d}(\sqrt{nd} + \sqrt{m} + \sqrt{2\gamma n\log n})}{\sqrt{n}}.
\end{align*}
\end{proof}

\subsection{Initialization}\label{ss:init}

The initialization step follows from an important result on singular value/vector perturbation~\cite{W72} by Wedin, which generalizes the $\sin\theta$ theorem for the eigenvector/eigenvalue perturbation of Hermitian matrices in~\cite{DK70}.
\begin{theorem}[The generalized $\sin\theta$ theorem, see (3.11) in~\cite{W72}]\label{thm:Wed}
For a symmetric matrix $\BA$ and its perturbed matrix $\BA_{E} = \BA +\BE$, let $(\BU,\BV,\BSigma)$ and $(\BU_E,\BV_E,\BSigma_E)$ be the top $k$ singular left/right vectors and values of $\BA$ and $\BA_E$ respectively. Then it holds that
%\begin{align}
%\|(\I - \BU_E\BU^{\top}_E)  \BU\BU^{\top} \|  &  \leq \frac{\alpha+\delta}{2\alpha+\delta} \|(\I-\BU_E\BU^{\top}_E)  \BE  \BV\BV^{\top}\| + \frac{\alpha}{2\alpha+\delta} \|  \BU\BU^{\top}\BE(\I-\BV_E\BV^{\top}_E)  \|, \label{eq:wedin1}\\
%\|(\I- \BV_{E}\BV^{\top}_{E})  \BV\BV^{\top} \|  & \leq \frac{\alpha+\delta}{2\alpha+\delta}  \|  \BU\BU^{\top}\BE(\I-\BV_E\BV^{\top}_E)  \| + \frac{\alpha}{2\alpha+\delta}\|(\I-\BU_E\BU^{\top}_E)  \BE  \BV\BV^{\top}\|, \label{eq:wedin2}
%\end{align}
\begin{align*}
\|(\I - \BU_E\BU^{\top}_E)  \BU\BU^{\top} \| & \leq \frac{1}{\delta}\cdot \max\{ \|(\I-\BU_E\BU^{\top}_E)  \BE  \BV\BV^{\top}\|,\|(\I - \BV_E\BV^{\top}_E)  \BE^{\top}  \BU\BU^{\top} \| \}, \\
\|(\I - \BV_E\BV^{\top}_E)  \BV\BV^{\top} \| & \leq \frac{1}{\delta}\cdot \max\{ \|(\I-\BU_E\BU^{\top}_E)  \BE  \BV\BV^{\top}\|,\|(\I - \BV_E\BV^{\top}_E)  \BE^{\top}  \BU\BU^{\top} \| \},
\end{align*}
where $\sigma_{\min}(\BU\BSigma\BV^{\top}) \geq \alpha + \delta$ and $\|\BA_E - \BU_E\BSigma_E\BV_E^{\top} \|\leq \alpha.$ 

\end{theorem}

It is a classical result that 
\begin{align*}
\| \BU_E-\BU\BR\| \leq 2\|(\I - \BU_E\BU^{\top}_E)  \BU\| = 2\|(\I - \BU_E\BU^{\top}_E)  \BU\BU^{\top}\|
\end{align*}
where $\BR=\PP(\BU^{\top}\BU_E).$ Sometimes, it is also convienient to use the following variant of the Wedin's theorem:
\begin{align}
\frac{1}{2}\min_{\BR\in\Od(d)}\| \BU_E - \BU\BR \|  &  \leq \frac{1}{\delta}\cdot\max\{ \|  \BE  \BV\|, \|  \BE^{\top}\BU  \| \} \leq \frac{\|\BE\|}{\delta}, \label{eq:wedin3}\\
\frac{1}{2}\min_{\BR\in\Od(d)}\| \BV_E - \BV\BR \|  & \leq \frac{1}{\delta}\cdot\max\{ \|  \BE  \BV\|, \|  \BE^{\top}\BU  \| \} \leq \frac{\|\BE\|}{\delta}. \label{eq:wedin4}
\end{align}

Now we will present the proof of Lemma~\ref{lem:init}.
\begin{proof}[\bf Proof of Lemma~\ref{lem:init} and Theorem~\ref{thm:spectral}]
The spectral method first computes the top $d$ left singular vectors of $\BD$ and round them to orthogonal matrices, as introduced in~\eqref{def:init}. Let $(\BU,\BV,\BSigma)$ be the top $d$ left/right singular vectors and singular values of $\BD\in\RR^{nd\times m}$ in~\eqref{def:Z} where $\BU\in\RR^{nd\times d}$, $\BSigma\in\RR^{d\times d}$, and $\BV\in\RR^{m\times d}$. Then it holds that
\[
\BU = \BD \BV \BSigma^{-1} = (\BZ\BA + \sigma\BW)\BV\BSigma^{-1}.
\]
Let the SVD of $\BA\in\RR^{d\times m}$ be 
$\BA = \BU_0\BSigma_0\BV_0^{\top}$
where $\BU_0\in\RR^{d\times d}$, $\BSigma_0\in\RR^{d\times d}$, and $\BV_0\in\RR^{m\times d}$ with $m\geq d+1.$ Then we know that the SVD of $\BZ\BA$ is
\[
\BZ\BA = \left( \frac{1}{\sqrt{n}}\BZ\BU_0\right)\cdot \sqrt{n} \BSigma_0\cdot \BV_0^{\top}.
\]
We treat $\BD = \BZ\BA + \sigma\BW$ as a noisy copy of $\BZ\BA$.

Our strategy is to approximate $\BU$ by using $\BD\BV_0\BR\BSigma^{-1}$ where $\BR = \PP(\BV_0^{\top}\BV).$
Then 
\begin{align*}
\BU - \BD\BV_0\BR\BSigma^{-1} & = \BD (\BV  - \BV_0\BR)\BSigma^{-1}  = (\BZ\BA + \sigma\BW) (\BV  - \BV_0\BR)\BSigma^{-1}.
\end{align*}
To estimate the $i$th block $\BU_i$ of $\BU$ by that of $ \BD\BV_0\BR\BSigma^{-1} $, we have
\begin{align*}
\|\BU_i - (\BA+\sigma\BW_i)\BV_0\BR\BSigma^{-1}\| & \leq \|(\BA + \sigma\BW_i) (\BV  - \BV_0\BR)\BSigma^{-1} \| \\
& \leq \left(\|\BA\| \|\BV - \BV_0\BR\| + \sigma \|\BW_i(\BV - \BV_0\BR)\| \right)\|\BSigma^{-1}\|
\end{align*}
where the $i$th block of $\BD\BV_0\BR\BSigma^{-1}$ is $(\BA+\sigma\BW_i)\BV_0\BR\BSigma^{-1}.$

We will show later in this section, i.e., Lemma~\ref{lem:init_supp}, that
\begin{align*}
\|\BSigma^{-1}\| & \leq \frac{2}{\sqrt{n}\sigma_{\min}(\BA)}, \qquad \|\BSigma\| \leq 2\sqrt{n}\|\BA\|\\
\|\BV-\BV_0\BR\| & \leq \frac{4\sigma (\sqrt{nd} + \sqrt{m} + \sqrt{2\gamma\log n})}{\sqrt{n}\sigma_{\min}(\BA)}, \\
\max_{1\leq i\leq n}\| \BW_i(\BV-\BV_0\BR) \| & \leq   \frac{3(\sqrt{m}+\sqrt{d}+\sqrt{2\gamma \log n})}{4}\max_{1\leq i\leq n}\|\BU_i\| + \frac{5(2\sqrt{d}+\sqrt{2\gamma\log n})}{2}.
\end{align*} 

With these estimations at hand, each block of $\BU - \BD\BV_0\BR\BSigma^{-1}$ is bounded by
\begin{align*}
& \|\BU_i - (\BA+\sigma\BW_i)\BV_0\BR\BSigma^{-1}\| \\
& \leq \left(\|\BA\| \|\BV - \BV_0\BR\| + \sigma \|\BW_i(\BV - \BV_0\BR)\| \right)\|\BSigma^{-1}\| \\
& \leq \frac{4\sigma (\sqrt{nd} + \sqrt{m} + \sqrt{2\gamma\log n})}{\sqrt{n}\sigma_{\min}(\BA)} \cdot \frac{2\kappa}{\sqrt{n}}  + \frac{3\sigma(\sqrt{m} + \sqrt{d} + \sqrt{2\gamma\log n})}{2\sqrt{n}\sigma_{\min}(\BA)}\cdot \max_{1\leq i\leq n}\|\BU_i\|\\
& \qquad + \frac{5\sigma(2\sqrt{d}+ \sqrt{2\gamma \log n})}{\sqrt{n}\sigma_{\min}(\BA)} \\
& < \frac{\sigma (\sqrt{nd} + \sqrt{m} + \sqrt{2\gamma n\log n})  }{\sqrt{n}\sigma_{\min}(\BA)} \left( \frac{3}{2} \max_{1\leq i\leq n}\|\BU_i\| + \frac{18\kappa}{\sqrt{n}} \right).
\end{align*}
Then we have
\begin{align}
\|\BU_i - \BA\BV_0\BR\BSigma^{-1}\| 
& \leq \|\BU_i - (\BA+\sigma\BW_i)\BV_0\BR\BSigma^{-1}\| + \sigma \|\BW_i\BV_0\BR\BSigma^{-1}\| \nonumber \\
& \leq \frac{\sigma (\sqrt{nd} + \sqrt{m} + \sqrt{2\gamma n\log n})  }{\sqrt{n}\sigma_{\min}(\BA)} \left( \frac{3}{2} \max_{1\leq i\leq n}\|\BU_i\| + \frac{22\kappa}{\sqrt{n}} \right)\label{eq:keyUmax}
\end{align}
where 
\begin{align*}
\|\BW_i\BV_0\BR\BSigma^{-1}\| & \leq \frac{2\|\BW_i\BV_0\|}{\sqrt{n}\sigma_{\min}(\BA)}\leq \frac{2(2\sqrt{d} + \sqrt{2\gamma\log n})}{\sqrt{n}\sigma_{\min}(\BA)} \\
& \leq \frac{(\sqrt{nd}+\sqrt{m}+\sqrt{2\gamma n\log n})}{\sqrt{n}\sigma_{\min}(\BA)}  \cdot \frac{4}{\sqrt{n}}
\end{align*}
holds with probability at least $1-n^{-\gamma+1}$, 
following from Corollary~\ref{cor:gauss} and the fact that $\BW_i\BV_0\in\RR^{d\times d}$ is a Gaussian random matrix.

\noindent{\bf Part 1: Proof of~\eqref{eq:init1}.}

Under the assumption on $\sigma$, we can see from~\eqref{eq:keyUmax} that 
\begin{align*}
\max_{1\leq i\leq n} \|\BU_i\| & \leq 2\|\BA\BV_0\BR\BSigma^{-1}\| \leq 2\|\BA\|\|\BSigma^{-1}\|\leq \frac{4\kappa}{\sqrt{n}}
\end{align*}
since~\eqref{eq:keyUmax} gives
\[
\max_{1\leq i\leq n}\|\BU_i\|  \leq \frac{24\sigma \kappa(\sqrt{nd} + \sqrt{m} + \sqrt{2\gamma n\log n})  }{\sqrt{n}\sigma_{\min}(\BA)} \max_{1\leq i\leq n}\|\BU_i\| + \|\BA\BV_0\BR\BSigma^{-1}\|
\]
and $\max_{1\leq i\leq n}\|\BU_i\| \geq 1/\sqrt{n}.$

Now we can conclude from~\eqref{eq:keyUmax} that
\begin{align}
\|\BU_i - \BA\BV_0\BR\BSigma^{-1}\|  &\leq \frac{28\kappa\sigma (\sqrt{nd} + \sqrt{m} + \sqrt{2\gamma n\log n})  }{n\sigma_{\min}(\BA)}.
\label{eq:Ui}
\end{align}

\vskip0.25cm

Note that $\BS^0 = \PP_n(\BU)$. To obtain a bound of $d_F(\BS^0,\BZ)$, it suffices to control $\|\PP(\BU_i) - \PP(\BA\BV_0\BR\BSigma^{-1})\|$ by using Lemma~\ref{lem:Pcon}. For the smallest singular value of $\BA\BV_0\BR\BSigma^{-1}$, we have
\[
\sigma_{\min}(\BA\BV_0\BR\BSigma^{-1}) = \sigma_{\min}(\BU_0\BSigma_0\BR\BSigma^{-1}) \geq\sigma_{\min}(\BSigma_0)\sigma_{\min}(\BSigma^{-1}) \geq \frac{\sigma_{\min}(\BA)}{2\sqrt{n}\|\BA\|} = \frac{1}{2\sqrt{n}\kappa}.
\]
By letting $\BQ=\PP(\BA\BV_0\BR\BSigma^{-1})$, Lemma~\ref{lem:Pcon} implies that
\begin{align*}
\|\BS_i^0 - \BQ\|& = \| \PP(\BU_i) - \PP(\BA\BV_0\BR\BSigma^{-1})\| \\
& \leq \frac{2}{\sigma_{\min}(\BA\BV_0\BR\BSigma^{-1})} \cdot \|\BU_i - \BA\BV_0\BR\BSigma^{-1}\|\\
& \leq \frac{112\kappa^2\sigma (\sqrt{nd} + \sqrt{m} + \sqrt{2\gamma n\log n})  }{\sqrt{n}\sigma_{\min}(\BA)} 
\end{align*}
which gives the proof of Theorem~\ref{thm:spectral}.

Therefore, we have
\[
\|\BS_i^0-\BQ\|\leq \eps ~~\Longrightarrow ~~d_F(\BS^0,\BZ)\leq \eps\sqrt{nd}
\]
under the assumption that
\begin{equation}\label{eq:sigma_init2}
\sigma\leq \frac{\eps\sqrt{d}\sigma_{\min}(\BA)}{112\kappa^2} \cdot \frac{\sqrt{n}}{\sqrt{d}(\sqrt{nd} + \sqrt{m} + \sqrt{2\gamma n\log n})}
\end{equation}
where $\eps\sqrt{d} = 1/(32\kappa^2).$

\vskip0.25cm

\noindent{\bf Part 2: Proof of~\eqref{eq:init2} and~\eqref{eq:init3}}

Now we need to prove that $d_F(\BS^0,\BS^{i,0})\leq \eps\sqrt{d}$ where $\BS^0 = \PP_n(\BU)$ and $\BS^{i,0} = \PP_n(\BU^{(i)}).$
From~\eqref{eq:Ui} and $\sigma_{\min}(\BA\BV_0\BR\BSigma^{-1}) \geq (2\sqrt{n}\kappa)^{-1}$, we have
\[
\sigma_{\min}(\BU_i) \geq \sigma_{\min}(\BA\BV_0\BR\BSigma^{-1}) - \|\BU_i -\BA\BV_0\BR\BSigma^{-1} \| \geq \frac{1}{4\sqrt{n}\kappa}
\]
under the assumption on $\sigma$ in~\eqref{eq:sigma_init2}.
As a result, applying Lemma~\ref{lem:Pncon} gives
\begin{align*}
d_F(\BS^0, \BS^{i,0}) & \leq 8\sqrt{n}\kappa \| \BU - \BU^{(i)}\BQ^{(i)}_U \|_F \\
& \leq  \frac{448\kappa^2\sigma (\sqrt{nd} +\sqrt{m}+ \sqrt{2\gamma n\log n}) }{\sqrt{n}\sigma_{\min}(\BA)}  \leq \eps\sqrt{d}
\end{align*}
where $\BQ_U^{(i)} =\PP( (\BU^{(i)})^{\top}\BU )$ under
\[
\sigma \leq \frac{\eps\sqrt{d}\sigma_{\min}(\BA)}{448\kappa^2}\cdot \frac{\sqrt{n}}{\sqrt{d}(\sqrt{nd} + \sqrt{m} + \sqrt{2\gamma n\log n})}.
\]
Here the bound on $\|\BU-\BU^{(i)}\BQ_U^{(i)}\|$ is given by Lemma~\ref{lem:init_supp}:
\begin{align*}
\| \BU - \BU^{(i)}\BQ^{(i)}_U \|_F & \leq \frac{12\sigma(\sqrt{m} + \sqrt{d} + \sqrt{2\gamma\log n})}{\sqrt{n}\sigma_{\min}(\BA)}\max_{1\leq i\leq n}\|\BU_i\|  + \frac{4\sigma(2\sqrt{d} + \sqrt{2\gamma\log n})}{\sqrt{n}\sigma_{\min}(\BA)} \\
%& \leq \frac{\sigma (\sqrt{nd} +\sqrt{m}+ \sqrt{2\gamma n\log n}) }{ \sqrt{n}\sigma_{\min}(\BA)} \left(12\max_{1\leq i\leq n}\|\BU_i\| + \frac{8}{\sqrt{n}}\right) \\
& \leq \frac{56\kappa\sigma (\sqrt{nd} +\sqrt{m}+ \sqrt{2\gamma n\log n}) }{n\sigma_{\min}(\BA)} 
\end{align*}
where $\max_{1\leq i\leq n}\|\BU_i\| \leq 4\kappa/\sqrt{n}.$

Finally, $\BS^0\in\mathcal{N}_{\xi,\infty}$ follows exactly from the argument of Step 3 in Lemma~\ref{lem:key1}: it says that if $d_F(\BS^0,\BS^{i,0})\leq \eps\sqrt{d}$, then $\BS^0\in\mathcal{N}_{\xi,\infty}$ with $\xi > 2(\eps+2).$
\end{proof}

\begin{lemma}\label{lem:init_supp}
Suppose $(\BU,\BV,\BSigma)$, $(\BU^{(i)}, \BV^{(i)},\BSigma^{(i)})$, and $(\BU_0,\BV_0,\BSigma_0)$ are the top $d$ left/right singular vectors and singular values of the data matrix $\BD = \BZ\BA + \sigma\BW\in\RR^{nd\times m}$, $\BD^{(i)} = \BZ\BA+\sigma\BW^{(i)}\in\RR^{nd\times m}$ (defined in~\eqref{def:CDi}), and $\BA\in\RR^{d\times m}$ respectively. Then conditioned on~\eqref{eq:gauss3a},~\eqref{eq:gauss3b}, and~\eqref{eq:gauss3c}, we have
\begin{align*}
\|\BSigma^{-1}\| & \leq \frac{2}{\sqrt{n}\sigma_{\min}(\BA)}, \qquad \|\BSigma\| \leq 2\sqrt{n}\|\BA\| \\
\|\BV-\BV_0\BR\| & \leq \frac{4\sigma (\sqrt{nd} + \sqrt{m} + \sqrt{2\gamma\log n})}{\sqrt{n}\sigma_{\min}(\BA)}, \\
\max_{1\leq i\leq n}\| \BW_i(\BV-\BV_0\BR) \| & \leq   \frac{3(\sqrt{m}+\sqrt{d}+\sqrt{2\gamma \log n})}{4}\max_{1\leq i\leq n}\|\BU_i\| + \frac{5(2\sqrt{d}+\sqrt{2\gamma\log n})}{2}, \\
\max_{1\leq i\leq n}\| \BU - \BU^{(i)}\BQ_U^{(i)} \| &  \leq \frac{12\sigma(\sqrt{m} + \sqrt{d} + \sqrt{2\gamma\log n})}{\sqrt{n}\sigma_{\min}(\BA)}\max_{1\leq i\leq n}\|\BU_i\|  + \frac{4\sigma(2\sqrt{d} + \sqrt{2\gamma\log n})}{\sqrt{n}\sigma_{\min}(\BA)}, 
\end{align*}
where $\BR = \PP(\BV_0^{\top}\BV)$ and $\BQ_U^{(i)} = \PP((\BU^{(i)})^{\top}\BU)$
with probability at least $1-O(n^{-\gamma+1})$
under the assumption
\[
\sigma \leq \frac{\sqrt{n}\sigma_{\min}(\BA)}{16(\sqrt{nd}+\sqrt{m}+\sqrt{2\gamma n\log n})}.
\] 
\end{lemma}

\begin{proof}[\bf Proof of Lemma~\ref{lem:init_supp}] We divide the proof into three parts. 

\noindent{\bf Part 1: Estimation of $\|\BSigma^{-1}\|$ and $\|\BSigma\|.$}

By Weyl's theorem, it holds that
\[
\|\sigma_i(\BSigma) -\sqrt{n} \sigma_i(\BA)\|\leq \sigma\|\BW\| \leq \frac{\sqrt{n}\sigma_{\min}(\BA)}{2}, \quad 1\leq i\leq d,
\]
where~\eqref{eq:gauss3a} gives $\|\BW\|\leq \sqrt{nd}+\sqrt{m}+\sqrt{2\gamma\log n}$.
Note that $\BSigma$ is a $d\times d$ positive semidefinite matrix consisting of the top $d$ singular values of $\BD$. We have $2^{-1}\sqrt{n}\sigma_{\min}(\BA)\I_d\preceq\BSigma\preceq 2\sqrt{n}\sigma_{\max}(\BA)\I_d$ and
\begin{equation}\label{eq:sigmamin2}
\|\BSigma^{-1}\| \leq \frac{2}{\sqrt{n}\sigma_{\min}(\BA)}, \qquad \|\BSigma\|\leq 2\sqrt{n}\|\BA\|.
\end{equation}

\noindent{\bf Part 2: Estimation of $\|\BV-\BV_0\BR\|.$}
We estimate $\|\BV - \BV_0\BR\|$ by using~\eqref{eq:wedin4}. Note that Weyl's theorem again implies
\[
\sigma_{d+1}(\BD) \leq \sigma\|\BW\|, \qquad \sigma_d(\BZ\BA)= \sqrt{n}\sigma_{\min}(\BA).
\]
We set $\alpha = \sigma\|\BW\|$ and $\delta = \sqrt{n}\sigma_{\min}(\BA) - \sigma\|\BW\|$ in~\eqref{eq:wedin4}. This results in
\begin{align*}
\|\BV - \BV_0\BR\| & \leq  \frac{2\sigma\max\{\|\BW\BV_0\|,n^{-1/2}\|\BW^{\top}\BZ\BU_0\|\}}{\sqrt{n}\sigma_{\min}(\BA) - \sigma\|\BW\|} 
\end{align*}
where $\BZ\BA$ has its SVD as $n^{-1/2}\BZ\BU_0\cdot n^{1/2}\BSigma_0\BV_0$ and $\BU_0\in\RR^{d\times d}$ is the left singular vectors of $\BA$. Since $\BW\BV_0\in\RR^{nd\times d}$ and $n^{-1/2}\BW^{\top}\BZ\BU_0\in\RR^{m\times d}$ are both Gaussian random matrix, Lemma~\ref{cor:gauss} gives
\begin{equation}
\|\BW\BV_0\|\leq \sqrt{nd} + \sqrt{d}+\sqrt{2\gamma\log n}, \qquad n^{-1/2}\|\BW^{\top}\BZ\BU_0\|\leq \sqrt{m} + \sqrt{d} +\sqrt{2\gamma\log n}
\end{equation}
with probability at least $1-2n^{-\gamma}.$ Thus combining these estimates with~\eqref{eq:sigmamin2} leads to
\[
\|\BV - \BV_0\BR\| \leq \frac{4\sigma (\sqrt{nd} + \sqrt{m} + \sqrt{2\gamma\log n})}{\sqrt{n}\sigma_{\min}(\BA)}.
\]

\noindent{\bf Part 3: Leave-one-out technique for $\|\BW_i(\BV-\BV_0\BR)\|$}

Note that $\BV-\BV_0\BR$ are not statistically independent of $\BW_i$ since $\BV$ is the top $d$ right singular vectors of $\BD = \BZ\BA + \sigma\BW$. We use the leave-one-out technique to approximate $\BV-\BV_0\BR$ by using the counterpart of $\BV^{(i)}$. Apparently, $(\BU^{(i)},\BV^{(i)})$ is independent of $\BW_i$ since $\BW^{(i)}$ does not contain $\BW_i.$

Denote
\[
\BQ^{(i)} = \PP((\BV^{(i)})^{\top}\BV), \qquad \BR^{(i)} = \PP(\BV_0^{\top}\BV^{(i)}).
\]
We decompose $\BV-\BV_0\BR$ into three components and estimate each of them individually:
\begin{align*}
\|\BW_i(\BV - \BV_0\BR)\| & \leq \|\BW_i(\BV - \BV^{(i)}\BQ^{(i)})\| + \|\BW_i(\BV^{(i)} - \BV_0\BR^{(i)})\| + \|\BW_i \BV_0(\BR^{(i)}\BQ^{(i)} -\BR)\| \\
& = : T_1 + T_2 + T_3.
%& \leq  \|\BW_i\|\|\BV - \BV^{(i)}\BQ^{(i)}\| + (\sqrt{d}+\sqrt{2\gamma n\log n})\|\BV^{(i)} - \BV_0\BR^{(i)}\|  + 2\|\BW_i \BV_0\|
\end{align*}
We start with $T_2$ and $T_3$ which are easier to estimate.
For $T_2$, Corollary~\ref{cor:gauss} implies that
\[
 \|\BW_i(\BV^{(i)} - \BV_0\BR^{(i)})\| \leq (2\sqrt{d}+\sqrt{2\gamma \log n}) \|\BV^{(i)} - \BV_0\BR^{(i)}\|, ~\forall 1\leq i\leq n,
\]
with probability least $1 - n^{-\gamma+1}$.  

Using the same argument in Part 1 gives an estimation of $\|\BV^{(i)}-\BV_0\BR^{(i)}\|$:  
\begin{equation}\label{eq:ViV0}
\|\BV^{(i)} - \BV_0\BR^{(i)}\| \leq \frac{4\sigma (\sqrt{nd} + \sqrt{m} + \sqrt{2\gamma\log n})}{\sqrt{n}\sigma_{\min}(\BA)}.
\end{equation}
holds with probability at least $1 - 2n^{-\gamma+1}$ where $\BR^{(i)} = \PP(\BV_0^{\top}\BV^{(i)}).$ 

Combining it with~\eqref{eq:ViV0}, we have
\begin{equation}\label{eq:T2}
T_2\leq \frac{4\sigma (2\sqrt{d}+\sqrt{2\gamma \log n})(\sqrt{nd} + \sqrt{m} + \sqrt{2\gamma\log n})}{\sqrt{n}\sigma_{\min}(\BA)}, 
\end{equation}
with probability at least $1-3n^{-\gamma+1}$. 

\vskip0.25cm

For $T_3,$ we have
\begin{equation}\label{eq:T3}
T_3 \leq 2\|\BW_i\BV_0\| \leq 2(2\sqrt{d} + \sqrt{2\gamma \log n}) 
\end{equation}
with probability at least $1-n^{-\gamma+1}$, following from Corollary~\ref{cor:gauss}.

\vskip0.25cm

Finally, we focus on $T_1$.  Since $T_1\leq \|\BW_i\| \|\BV - \BV^{(i)}\BQ^{(i)}\|$, it suffices to control $\|\BV - \BV^{(i)}\BQ^{(i)}\|$. 
For the distance between $\BV$ and $\BV^{(i)}$, we apply Wedin's theorem by treating $\BD$ as a perturbation of $\BD^{(i)}$ in~\eqref{def:CDi}. In fact, the same argument to the distance between the left singular vectors $\BU$ and $\BU^{(i)}$ of $\BD$ and $\BD^{(i)}.$ Then we know that
\[
\sigma_{d+1}(\BD)\leq \sigma\|\BW\|, \qquad \sigma_d(\BD^{(i)})\geq \sqrt{n}\sigma_{\min}(\BA) - \sigma\|\BW\|
\]
and thus it is natural to set $\alpha = \sigma\|\BW\|$ and $\delta = \sqrt{n}\sigma_{\min}(\BA) - 2\sigma\|\BW\|$ in Theorem~\ref{thm:Wed}. Then~\eqref{eq:wedin3} and~\eqref{eq:wedin4} give
\begin{align*}
& \max\left\{ \|\BV - \BV^{(i)}\BQ^{(i)}\|, \|\BU - \BU^{(i)}\BQ^{(i)}_U\| \right\} \\
& \qquad  \leq \frac{2\sigma \max\{\| (\BW - \BW^{(i)})\BV^{(i)} \|, \| (\BW - \BW^{(i)})^{\top}\BU^{(i)} \|\}}{ \sqrt{n}\sigma_{\min}(\BA) - 2\sigma\|\BW\|}  \\
& \qquad \leq \frac{4\sigma (\| \BW_i\BV^{(i)} \|+\| \BW_i^{\top}\BU_i^{(i)} \|)}{ \sqrt{n}\sigma_{\min}(\BA)} % \\
\end{align*}
where $\BQ^{(i)}_U = \PP((\BU^{(i)})^{\top}\BU)$ is orthogonal and $\sqrt{n}\sigma_{\min}(\BA)\geq 4\sigma\|\BW\|$.

Note that $\BV^{(i)}$ is a $m\times d$ matrix with orthogonal columns and $\BU_i^{(i)}\in\RR^{d\times d}$ is the $i$th $d\times d$ block of $\BU^{(i)}$. With probability at least $1-2n^{-\gamma+1}$, Corollary~\ref{cor:gauss} gives
\[
\|\BW_i\BV^{(i)}\| \leq 2\sqrt{d} + \sqrt{2\gamma\log n}, \quad \| \BW_i^{\top}\BU_i^{(i)} \| \leq (\sqrt{m} + \sqrt{d} + \sqrt{2\gamma\log n})\|\BU^{(i)}_i\|
\]
for all $1\leq i\leq n.$ As a result, under the assumption on $\sigma,$ we have 
\begin{align}
& \max\left\{ \|\BV - \BV^{(i)}\BQ^{(i)}\|, \|\BU - \BU^{(i)}\BQ^{(i)}_U\| \right\} \nonumber \\
& \leq \frac{4\sigma}{\sqrt{n}\sigma_{\min}(\BA)}\left( (\sqrt{m} + \sqrt{d} + \sqrt{2\gamma\log n})\|\BU^{(i)}_i\|  +  (2\sqrt{d} + \sqrt{2\gamma\log n}) \right). \label{eq:VVQ1}
\end{align}

From the equation above, we have
\[
\Big| \| \BU^{(i)}_i \| - \|\BU_i\| \Big|  \leq \frac{4\sigma}{\sqrt{n}\sigma_{\min}(\BA)}\left( (\sqrt{m} + \sqrt{d} + \sqrt{2\gamma\log n})\|\BU^{(i)}_i\|  +  (2\sqrt{d} + \sqrt{2\gamma\log n}) \right).
\]
Under $\sigma \leq \sqrt{n}\sigma_{\min}(\BA)/(16(\sqrt{nd}+\sqrt{m}+\sqrt{2\gamma n \log n}))$, the equation above leads to
\[
\|\BU_i^{(i)}\| \leq 2\|\BU_i\|+ \frac{8\sigma(2\sqrt{d} + \sqrt{2\gamma\log n})}{\sqrt{n}\sigma_{\min}(\BA)} \leq 2\|\BU_i\| + \frac{1}{\sqrt{n}}.
\]
Note that $\BU^{\top}\BU = \sum_{i=1}^n \BU_i^{\top}\BU_i = \I_d$, and then $\max\|\BU_i\| \geq 1/\sqrt{n}.$ Then we have
\[
\|\BU_i^{(i)}\| \leq 3\max_{1\leq i\leq n}\|\BU_i\|.
\]

By substituting the estimation of $\|\BU_i^{(i)}\|$ into~\eqref{eq:VVQ1}, we have an upper bound of $T_1$ as well as $\|\BU - \BU^{(i)}\BQ_U^{(i)}\|:$
\begin{align}
& \max\left\{ \|\BV - \BV^{(i)}\BQ^{(i)}\|, \|\BU - \BU^{(i)}\BQ^{(i)}_U\| \right\} \nonumber \\
& \qquad \leq \frac{4\sigma}{\sqrt{n}\sigma_{\min}(\BA)}\left(3 (\sqrt{m} + \sqrt{d} + \sqrt{2\gamma\log n})\max_{1\leq i\leq n}\|\BU_i\|  +  (2\sqrt{d} + \sqrt{2\gamma\log n}) \right) \label{eq:VVQ}
\end{align}
and
\begin{align}%\label{eq:T1}
T_1 & \leq \|\BW_i\|\|\BV-\BV^{(i)}\BQ^{(i)}\| = (\sqrt{m}+\sqrt{d}+\sqrt{2\gamma\log n}) \|\BV-\BV^{(i)}\BQ^{(i)}\| . \label{eq:T1}
\end{align}

As a result, combining~\eqref{eq:T1},~\eqref{eq:T2}, and~\eqref{eq:T3} together gives
\begin{align*}
& \max_{1\leq i\leq n}\|\BW_i(\BV - \BV_0\BR)\|  \leq T_1 + T_2 + T_3 \\
& \leq \frac{12\sigma(\sqrt{m} + \sqrt{d} + \sqrt{2\gamma\log n})^2}{\sqrt{n}\sigma_{\min}(\BA)} \max_{1\leq i\leq n}\|\BU_i\|  +  \frac{4\sigma(2\sqrt{d} + \sqrt{2\gamma\log n})(\sqrt{m} + \sqrt{d} + \sqrt{2\gamma\log n})}{\sqrt{n}\sigma_{\min}(\BA)}  \\
&  \qquad +  \frac{4\sigma (2\sqrt{d}+\sqrt{2\gamma \log n})(\sqrt{nd} + \sqrt{m} + \sqrt{2\gamma\log n})}{\sqrt{n}\sigma_{\min}(\BA)} + 2(2\sqrt{d} + \sqrt{2\gamma\log n}) \\
 & \leq  \frac{3(\sqrt{m}+\sqrt{d}+\sqrt{2\gamma \log n})}{4}\max_{1\leq i\leq n}\|\BU_i\| + \frac{5(2\sqrt{d}+\sqrt{2\gamma\log n})}{2}
\end{align*}
where $\sigma \leq \sqrt{n}\sigma_{\min}(\BA)/(16(\sqrt{nd}+\sqrt{m}+\sqrt{2\gamma n \log n}))$.
\end{proof}

%\bibliography{GOPP.bib}
%\bibliographystyle{abbrv}

\end{document}